\documentclass[11pt,notitlepage]{revtex4}

\usepackage{pifont}
\usepackage{mathpazo}

\usepackage{amsmath}
\usepackage[T1]{fontenc}
\usepackage[utf8]{inputenc}
\usepackage{graphicx,epic,eepic,epsfig,amsmath,latexsym,amssymb,verbatim,color}
\usepackage{dsfont}

 \usepackage{tikz}

 \usepackage{hyperref}
\hypersetup{colorlinks=true,citecolor=blue,linkcolor=blue,filecolor=blue,urlcolor=blue,breaklinks=true}

\usepackage[marginal]{footmisc}
\usepackage{url}
\usepackage{theorem}
\newtheorem{definition}{Definition}
\newtheorem{proposition}[definition]{Proposition}

\newtheorem{theorem}[definition]{Theorem}
\newtheorem{corollary}[definition]{Corollary}

\def\squareforqed{\hbox{\rlap{$\sqcap$}$\sqcup$}}
\def\qed{\ifmmode\squareforqed\else{\unskip\nobreak\hfil
\penalty50\hskip1em\null\nobreak\hfil\squareforqed
\parfillskip=0pt\finalhyphendemerits=0\endgraf}\fi}
\def\endenv{\ifmmode\;\else{\unskip\nobreak\hfil
\penalty50\hskip1em\null\nobreak\hfil\;
\parfillskip=0pt\finalhyphendemerits=0\endgraf}\fi}
\newenvironment{proof}{\noindent \textbf{{Proof~} }}{\qed}
\newenvironment{remark}{\noindent \textbf{{Remark~}}}{\qed}

\mathchardef\ordinarycolon\mathcode`\:
\mathcode`\:=\string"8000
\def\vcentcolon{\mathrel{\mathop\ordinarycolon}}
\begingroup \catcode`\:=\active
  \lowercase{\endgroup
  \let :\vcentcolon
  }

\newcommand{\nc}{\newcommand}
\nc{\rnc}{\renewcommand}
\nc{\beg}{\begin{equation}}
\nc{\eeq}{{\end{equation}}}
\nc{\beqa}{\begin{eqnarray}}
\nc{\eeqa}{\end{eqnarray}}
\nc{\lbar}[1]{\overline{#1}}
\nc{\bra}[1]{\langle#1|}
\nc{\ket}[1]{|#1\rangle}
\nc{\ketbra}[2]{|#1\rangle\!\langle#2|}
\nc{\braket}[2]{\langle#1|#2\rangle}

\nc{\proj}[1]{| #1\rangle\!\langle #1 |}
\nc{\avg}[1]{\langle#1\rangle}
\nc{\Rank}{\operatorname{Rank}}
\nc{\smfrac}[2]{\mbox{$\frac{#1}{#2}$}}
\nc{\tr}{\operatorname{Tr}}
\nc{\ox}{\otimes}
\nc{\dg}{\dagger}
\nc{\dn}{\downarrow}
\nc{\cA}{{\cal A}}
\nc{\cB}{{\cal B}}
\nc{\cC}{{\cal C}}
\nc{\cD}{{\cal D}}
\nc{\cE}{{\cal E}}
\nc{\cF}{{\cal F}}
\nc{\cG}{{\cal G}}
\nc{\cH}{{\cal H}}
\nc{\cI}{{\cal I}}
\nc{\cJ}{{\cal J}}
\nc{\cK}{{\cal K}}
\nc{\cL}{{\cal L}}
\nc{\cM}{{\cal M}}
\nc{\cN}{{\cal N}}
\nc{\cO}{{\cal O}}
\nc{\cP}{{\cal P}}
\nc{\cQ}{{\cal Q}}
\nc{\cR}{{\cal R}}
\nc{\cS}{{\cal S}}
\nc{\cT}{{\cal T}}
\nc{\cV}{{\cal V}}
\nc{\cX}{{\cal X}}
\nc{\cY}{{\cal Y}}
\nc{\cZ}{{\cal Z}}
\nc{\cW}{{\cal W}}
\nc{\csupp}{{\operatorname{csupp}}}
\nc{\qsupp}{{\operatorname{qsupp}}}
\nc{\var}{{\operatorname{var}}}
\nc{\rar}{\rightarrow}
\nc{\lrar}{\longrightarrow}
\nc{\polylog}{{\operatorname{polylog}}}
\nc{\1}{{\mathds{1}}}
\nc{\wt}{{\operatorname{wt}}}
\nc{\av}[1]{{\left\langle {#1} \right\rangle}}
\nc{\supp}{{\operatorname{supp}}}

\def\a{\alpha}
\def\b{\beta}

\def\e{\epsilon}

\def\o{\omega}

\def\U{\Upsilon}

\def\O{\Omega}

\nc{\RR}{{{\mathbb R}}}
\nc{\CC}{{{\mathbb C}}}
\nc{\FF}{{{\mathbb F}}}
\nc{\NN}{{{\mathbb N}}}
\nc{\ZZ}{{{\mathbb Z}}}
\nc{\PP}{{{\mathbb P}}}
\nc{\QQ}{{{\mathbb Q}}}
\nc{\UU}{{{\mathbb U}}}
\nc{\EE}{{{\mathbb E}}}
\nc{\id}{{\operatorname{id}}}

\nc{\CHSH}{{\operatorname{CHSH}}}

\nc{\be}{\begin{equation}}
\nc{\ee}{{\end{equation}}}
\nc{\bea}{\begin{eqnarray}}
\nc{\eea}{\end{eqnarray}}
\nc{\<}{\langle}
\rnc{\>}{\rangle}
\nc{\Hom}[2]{\mbox{Hom}(\CC^{#1},\CC^{#2})}
\nc{\rU}{\mbox{U}}

\nc{\ob}[1]{#1}

\nc{\SEP}{{\text{SEP}}}
\nc{\NS}{{\text{NS}}}
\nc{\LOCC}{{\text{LOCC}}}
\nc{\PPT}{{\text{PPT}}}
\nc{\EXT}{{\text{EXT}}}
\nc{\Sym}{{\operatorname{Sym}}}

\nc{\ERLO}{{E_{\text{r,LO}}}}
\nc{\ERLOCC}{{E_{\text{r,LOCC}}}}
\nc{\ERPPT}{{E_{\text{r,PPT}}}}
\nc{\ERLOCCinfty}{{E^{\infty}_{\text{r,LOCC}}}}
\nc{\Aram}{{\operatorname{\sf A}}}

\begin{document}
\title{Semidefinite programming strong converse bounds for classical capacity}

\author{Xin Wang$^{1}$}
\email{xin.wang-8@student.uts.edu.au}
\author{Wei Xie$^{1}$}
\email{xievvvei@gmail.com}
\author{Runyao Duan$^{1,2}$}
\email{runyao.duan@uts.edu.au}
\affiliation{$^1$Centre for Quantum Software and Information, Faculty of Engineering and Information Technology, University of Technology Sydney, NSW 2007, Australia}
\affiliation{$^2$UTS-AMSS Joint Research Laboratory for Quantum Computation and Quantum Information Processing, Academy of Mathematics and Systems Science, Chinese Academy of Sciences, Beijing 100190, China}
\thanks{A preliminary version of this paper was presented at the 20th Annual Conference on Quantum Information Processing and the IEEE International Symposium on Information Theory in 2017 \cite{Wang2017b}.}

\begin{abstract}
We investigate the classical communication over quantum channels when assisted by no-signalling (NS) and PPT-preserving (PPT) codes, for which both the optimal success probability of a given transmission rate and the one-shot $\epsilon$-error capacity are formalized as semidefinite programs (SDPs). Based on this, we obtain improved SDP finite blocklength converse bounds of general quantum channels for entanglement-assisted codes and unassisted codes. Furthermore, we derive two SDP strong converse bounds for the classical capacity of general quantum channels: for any code with a rate exceeding either of the two bounds of the channel, the success probability vanishes exponentially fast as the number of channel uses increases. In particular, applying our efficiently computable bounds, we derive an improved upper bound on the classical capacity of the amplitude damping channel. We also establish the strong converse property for the classical and private capacities of a new class of quantum channels. We finally study the zero-error setting and provide efficiently computable upper bounds on the one-shot zero-error capacity of a general quantum channel.
\end{abstract}

\maketitle

\section{Introduction}
The reliable transmission of classical information via noisy quantum channels is central to quantum information theory. 
The classical capacity of a noisy quantum channel is the highest rate at which it can convey classical   information reliably over asymptotically many uses of the channel. 
The Holevo-Schumacher-Westmoreland (HSW) theorem \cite{Holevo1973,Holevo1998,Schumacher1997} gives a full characterization of the classical capacity of quantum channels:
\begin{equation}\label{CN}
C(\cN):=\sup_{n\ge 1}\frac{\chi(\cN^{\ox n})}{n}, \\
\end{equation}
where $\chi(\cN)$ is the Holevo capacity of the channel $\cN$ given by 
$\chi(\cN):=\max_{\{(p_i,\rho_i)\}} H\left(\sum_ip_i\cN(\rho_i)\right)-\sum_ip_iH(\cN(\rho_i))$, $\{(p_i,\rho_i)\}_i$ is an ensemble of quantum states on $A$ and $H(\sigma)=-\tr \sigma\log\sigma$ is the von Neumann entropy of a quantum state. Throughout this paper, $\log$ denotes the binary logarithm.

For certain classes of quantum channels  (depolarizing channel \cite{King2003}, erasure channel \cite{Bennett1997}, unital qubit channel \cite{King2002}, etc. \cite{Amosov2000,Datta2006,Fukuda2005,Konig2012}), the classical capacity of the channel is equal to the Holevo capacity, since their Holevo capacities are all additive. However, for a general quantum channel, our understanding of the classical capacity is still limited. The work of Hastings  \cite{Hastings2008a} shows that the Holevo capacity is generally not additive, and thus the regularization in Eq. (\ref{CN}) is necessary in general.
Since the complexity of computing the Holevo capacity is NP-complete \cite{Beigi2007}, the regularized Holevo capacity of a general quantum channel is notoriously difficult to calculate.
Even for the qubit amplitude damping channel, the classical capacity remains unknown.

The converse part of the HSW theorem states that if the communication rate exceeds the capacity, then the error probability of any coding scheme cannot approach zero in the limit of many channel uses. 
This kind of  ``weak'' converse suggests the possibility for one to increase communication rates by allowing an increased error probability.
A   {\em strong converse property} leaves no such room for the trade-off; i.e.,  the error probability necessarily converges to one in the limit of many channel uses whenever the rate exceeds the capacity of the channel. For classical channels, the strong converse property for the classical capacity was established by  Wolfowitz \cite{Wolfowitz1978}. 
For quantum channels, the strong converse property for the classical capacity has been confirmed for several classes of channels \cite{Ogawa1999,Winter1999,Koenig2009,Wilde2013a,Wilde2014a}. 
 Winter \cite{Winter1999} and Ogawa and Nagaoka \cite{Ogawa1999} independently established the strong converse property for the classical capacity of classical-quantum channels. Koenig and Wehner \cite{Koenig2009} proved the strong converse property for particular covariant quantum channels. Recently, for the entanglement-breaking and Hadamard channels, the strong converse property was proved by Wilde, Winter and Yang \cite{Wilde2014a}. Moreover, the strong converse property for the pure-loss bosonic channel was proved by Wilde and Winter \cite{Wilde2013a}.
Unfortunately, for a general quantum channel, less is known about the strong converse property of the classical capacity, and it remains open whether this property holds for all quantum channels. 
A {\em strong converse bound} for the classical capacity is a quantity such that the success probability of transmitting classical messages vanishes exponentially fast as the number of channel uses increases if the rate of communication exceeds this quantity, which forbids the trade-off between rate and error in the limit of many channel uses.

Another fundamental problem, of both theoretical and practical interest, is the trade-off between the channel uses, communication rate and error probability in the non-asymptotic (or finite blocklength) regime. In a realistic setting, the number of channel uses is necessarily limited in quantum information processing. Therefore one has to make a trade-off between the transmission rate and error tolerance. 
Note that one only needs to study one-shot communication over the channel since it can correspond to a finite blocklength and one can also study the asymptotic capacity via the finite blocklength approach. The study of  finite blocklength regime has recently garnered great interest in classical information theory (e.g., \cite{Polyanskiy2010,Hayashi2009,Matthews2012}) as well as in quantum information theory (e.g., \cite{Matthews2014,Wang2012,Renes2011,Tomamichel2013a,Berta2011a,Leung2015c,Tomamichel2015,Beigi2015,Tomamichel2015b,Tomamichel2016,Fang2017,Cheng2017b,Chubb2017}). For classical channels, Polyanskiy, Poor, and Verd\'u \cite{Polyanskiy2010} derive the finite blocklength converse bound via hypothesis testing and Matthews \cite{Matthews2012} provides an alternative proof of this converse bound via classical no-signalling codes. For  classical-quantum channels, the one-shot converse and achievability bounds are given in \cite{Mosonyi2009a,Wang2012,Renes2011}. 
Recently, the one-shot converse bounds for entanglement-assisted and unassisted codes were given in \cite{Matthews2014}, which generalizes the hypothesis testing approach in \cite{Polyanskiy2010} to quantum channels.

To gain insights into the generally intractable problem of evaluating the capacities of quantum channels,  a natural approach is to study the performance of extra free resources in the coding scheme. This scheme, called a {\em code}, is equivalently a bipartite operation performed jointly by the sender Alice and the receiver Bob to assist the communication \cite{Leung2015c}.  The {\em PPT-preserving codes}, i.e. the PPT-preserving bipartite operations, include all operations that can be implemented by local operations and classical communication (LOCC) and were introduced to study entanglement distillation in an early paper by Rains \cite{Rains2001}. 
The {\em no-signalling (NS) codes} refer to the bipartite quantum operations with the no-signalling constraints, which arise in the research of the relativistic causality of quantum operations \cite{Beckman2001, Eggeling2002a, Piani2006, Oreshkov2012}. Recently these general codes have been used to study the zero-error classical communication \cite{Duan2016} and quantum communication \cite{Leung2015c} over quantum channels. 
Our work follows this approach and focuses on classical  communication via quantum
channels assisted by  NS and  NS$\cap$PPT codes. 

\section{Summary of results}
In this paper, we focus on the reliable classical communication over quantum channels assisted by no-signalling and PPT-preserving codes under both non-asymptotic (or finite blocklength) and asymptotic settings.  The summary of our results is as follows.

In Section \ref{NS PPT finite},  we formalize the optimal average success probability of transmitting classical messages over a quantum channel assisted by NS or NS$\cap$PPT codes as SDPs. Using these SDPs,  we establish the one-shot  NS-assisted (or NS$\cap$PPT-assisted) $\e$-error capacity, i.e., the maximum rate of classical communication with a fixed error threshold.  We further compare these one-shot $\e$-error capacities with the previous SDP-computable entanglement-assisted (or unassisted) converse bound derived by the technique of quantum hypothesis testing in \cite{Matthews2014}. Our one-shot $\e$-error capacities, which consider potentially stronger assistances, are always no larger than the previous SDP bounds, and the inequalities can be strict even for qubit channels or classical-quantum channels.  This means that our one-shot $\e$-error capacities can provide tighter finite blocklength converse bounds for the entanglement-assisted and unassisted classical capacity. Moreover, our one-shot $\e$-error capacities also reduce to the Polyanskiy-Poor-Verd\'u (PPV) converse bound  \cite{Polyanskiy2010} for classical channels.
Furthermore, in common with the quantum hypothesis testing converse bound \cite{Matthews2014} and the bound of Datta and Hsieh \cite{Datta2013c}, the large block length behaviour of our one-shot NS-assisted  $\e$-error capacity also recovers the converse part of the formula for entanglement-assisted capacity \cite{Bennett1999} and implies that no-signalling-assisted classical capacity coincides with the entanglement-assisted classical capacity.

In Section \ref{NS PPT asymptotics}, we derive two SDP strong converse bounds for the  NS$\cap$PPT-assisted classical capacity of a general quantum channel based on the one-shot characterization of the optimal success probability. These bounds also provide efficiently computable strong converse bounds for the classical capacity.  As a special case, we show that $\log(1+\sqrt{1-\gamma})$ is a strong converse bound for the classical capacity of the amplitude damping channel with parameter $\gamma$, and this improves the best previously known upper bound in  \cite{Brandao2011c}. Furthermore, applying our strong converse bounds, we also prove the strong converse property for the classical and private capacities of a new class of quantum channels.

In Section \ref{NS PPT zero error}, we consider the zero-error communication problem \cite{Shannon1956},  which requires that the communication is with zero probability of error. To be specific, based on our SDPs of optimal success probability, we derive the one-shot NS-assisted (or NS$\cap$PPT-assisted)  zero-error capacity of general quantum channels. Our result of the NS-assisted capacity provides an alternative proof of the NS-assisted zero-error capacity in \cite{Duan2016}. Moreover, our one-shot NS$\cap$PPT-assisted zero-error capacity gives an SDP-computable upper bound on the one-shot unassisted zero-error capacity, and it can be strictly smaller than the previous upper bound in \cite{Duan2013}.

Finally, in Section \ref{conclusion}, we make a conclusion and leave some interesting open questions.
\section{Preliminaries}
In the following, we will frequently use symbols such as $A$ (or $A'$) and $B$ (or $B'$) to denote (finite-dimensional) Hilbert spaces associated with Alice and Bob, respectively. We use $d_A$ to denote the dimension of system $A$. The set of linear operators over $A$ is denoted by $\cL(A)$. We usually write an operator with subscript indicating the system that the operator acts on, such as $T_{AB}$, and write $T_A:=\tr_B T_{AB}$. Note that for a linear operator $R\in\cL(A)$, we define $|R|=\sqrt{R^\dagger R}$, where $R^\dagger$ is the conjugate transpose of $R$, and the trace norm of $R$ is given by $\|R\|_1=\tr |R|$. The operator norm $\|R\|_\infty$ is defined as the maximum eigenvalue of $|R|$. 
A deterministic quantum operation (quantum channel) $\cN$ ($A'\to B$) is simply a completely positive (CP) and trace-preserving (TP) linear map from $\cL(A')$ to $\cL(B)$.  The Choi-Jamio\l{}kowski matrix \cite{Jamiokowski1972,Choi1975}  of $\cN$ is given by $J_{\cN}=\sum_{ij} \ketbra{i_A}{j_{A}} \ox \cN(\ketbra{i_{A'}}{j_{A'}})$, where $\{\ket{i_A}\}$ and $\{\ket{i_{A'}}\}$ are orthonormal bases on isomorphic Hilbert spaces $A$ and $A'$, respectively. 
 A positive semidefinite operator $E \in \cL(A \ox B)$ is
said to be a positive
partial transpose operator (or simply PPT) if $E^{T_{B}}\geq 0$, where ${T_{B}}$ means the partial transpose with respect to the party
$B$, i.e., $(\ketbra{ij}{kl})^{T_{B}}=\ketbra{il}{kj}$.
 As shown in \cite{Rains2001},
a bipartite operation $\Pi(A_i B_i\to A_o B_o)$ is PPT-preserving if and only if its Choi-Jamio\l{}kowski matrix $Z_{A_iB_iA_oB_o}$  is PPT. We sometimes omit the identity operator or operation $\1$, for example, $\cE(A\to B)(X_{AC})\equiv (\cE(A\to B)\ox\1_C)(X_{AC})$.

The constraints of PPT and NS can be mathematically characterized as follows.  A bipartite operation $\Pi(A_iB_i\to A_oB_o)$ is no-signalling and PPT-preserving if and only if its Choi-Jamio\l{}kowski matrix $Z_{A_iB_iA_oB_o}$ satisfies \cite{Leung2015c}:
\begin{equation}\label{code constraints}
\begin{split}
Z_{A_iB_iA_oB_o}\ge 0,  &\quad (\text{CP})\\
Z_{A_iB_i} = \1_{A_i B_i}, &\quad (\text{TP})\\
Z_{A_iB_iA_oB_o}^{T_{B_iB_o}}\ge 0, &\quad(\text{PPT}) \\
Z_{A_iB_iB_o}=\frac{\1_{A_i}}{d_{A_i}}\ox Z_{B_iB_o}, &\quad (A \not\rightarrow B)\\
Z_{A_iB_iA_o}=\frac{\1_{B_i}}{d_{B_i}}\ox Z_{A_iA_o}, &\quad (B \not\rightarrow A) \\
\end{split}\end{equation}
where the five lines correspond to characterize that $\Pi$ is completely positive, trace-preserving, PPT-preserving, no-signalling from A to B, no-signalling from B to A, respectively. The structure of no-signalling codes is also studied in \cite{Duan2016}.

Semidefinite programming  \cite{Vandenberghe1996} is a subfield of convex optimization and is a  powerful tool in quantum information theory with many applications (e.g., \cite{Matthews2014,Leung2015c,Duan2016,Rains2001,Wang2016,Harrow2015,Wang2016c,Li2017,Berta2015,Xie2017}).
There are known polynomial-time algorithms for semidefinite programming \cite{Khachiyan1980}. 
In this work, we use the CVX software (a Matlab-based convex modeling framework)
\cite{Grant2008}  and QETLAB (A Matlab Toolbox for Quantum
Entanglement) \cite{NathanielJohnston2016}
to solve the SDPs. 
Details about  semidefinite programming can be found in \cite{Watrous2011b}.

\section{Classical communication assisted by NS and PPT codes}
\label{NS PPT finite}
\subsection{Semidefinite programs for optimal success probability}
Suppose Alice wants to send the classical message labeled by $\{1,\dots,m\}$ to Bob using the composite channel $\cM=\Pi\circ\cN$, where $\Pi$ is a bipartite operation that generalizes the usual encoding scheme $\cE$ and decoding scheme $\cD$, see Fig. 1 for details. 
In this paper, we consider $\Pi$ as the bipartite operation implementing the $\rm{NS\cap PPT}$ or $\rm{NS}$ assistance. 
After the action of $\cE$ and $\cN$, the message results in quantum state at Bob's side. Bob then performs a POVM with $m$ outcomes on the resulting quantum state. The POVM is a component of the operation $\cD$. Since the results of the POVM and the input messages are both classical, it is natural to assume that $\cM$ is with classical registers throughout this paper, that is, $\Delta\circ\cM\circ\Delta=\cM$ for some completely dephasing channel $\Delta$. If the outcome $k\in\{1,\dots,m\}$ happens, he concludes that the message with label $k$ was sent. Let $\O$ be some class of bipartite operations. The average success probability of the general code $\Pi$ and the $\O$-class code is defined as follows.

\begin{definition}
The average success probability of $\cN$ to transmit $m$ messages assisted with the code $\Pi$ is defined 
by
\begin{equation}
f(\cN,\Pi,m)=\frac{1}{m}\sum_{k=1}^{m}\tr(\cM(\proj k)\proj{k}),
\end{equation}
where $\cM\equiv\Pi\circ\cN$ and $\{\ket{k}\}$ is the computational basis in system $A_i$.

Furthermore, the optimal average success probability of $\cN$ to transmit $m$ messages assisted with $\O$-class code is defined by
\begin{equation}
f_\O(\cN,m)= \sup_{\Pi} f(\cN,\Pi,m),
\end{equation}
where the maximum is over the codes in class $\O$.
\end{definition}

We now define the $\O$-assisted classical capacity of a quantum channel as follows.
\begin{definition}
\begin{equation}
        C_{\O}(\cN)
        := \sup \left\{r: \mathop {\lim }\limits_{n \to \infty }  
        f_\O(\cN^{\ox n}, 2^{rn}) = 1\right\}.
\end{equation}
\end{definition}

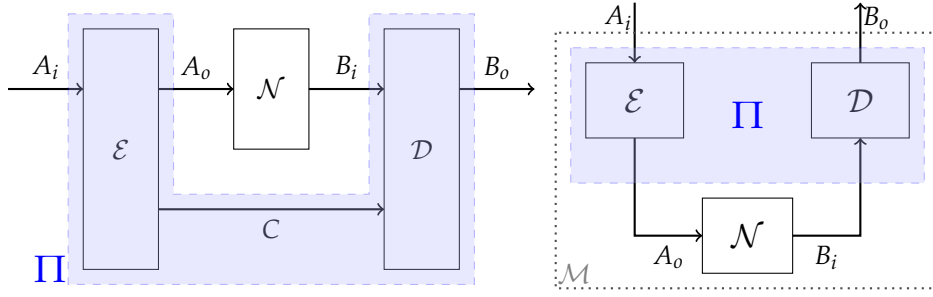
\begin{figure}
\centering
\begin{tikzpicture}
    \def\xa{1};\def\xb{2};\def\xc{3};\def\xd{4};\def\xe{5};\def\xf{6};\def\xg{7};
    \def\ya{0.8};\def\yc{-1.6};\def\yd{-2.4};  \def\o{0.2};
    \pgfmathsetmacro\yb{-\ya}
    \draw[thick,->] (0,0) -- node[above] {$A_i$} (\xa,0);
    \draw[thick,->] (\xb,0) -- node[above] {$A_o$} (\xc,0);
    \draw[thick,->] (\xd,0) -- node[above] {$B_i$} (\xe,0);
    \draw[thick,->] (\xf,0) -- node[above] {$B_o$} (\xg,0);
    \draw[thick,->] (\xb,\yc) -- node[below] {$C$} (\xe,\yc);
    \draw (\xa,\ya) rectangle (\xb,\yd) node[midway] {$\cE$};
    \draw (\xc,\ya) rectangle (\xd,\yb) node[midway] {$\cN$};
    \draw (\xe,\ya) rectangle (\xf,\yd) node[midway] {$\cD$};
    \draw[dashed,blue,fill=blue!30,opacity=0.3] (\xa-\o,\ya+\o) -- (\xb+\o,\ya+\o) -- (\xb+\o,\yc+\o) -- (\xe-\o,\yc+\o)-- (\xe-\o,\ya+\o)  -- (\xf+\o,\ya+\o) -- (\xf+\o,\yd-\o) -- 
    (\xa-\o,\yd-\o) -- (\xa-\o,\ya+\o);
    \node[blue,left,shift={(0.1,0.2)}] (d) at (\xa-\o,\yd-\o) {\Large $\Pi$};
\end{tikzpicture}
  \label{fig:test2}
\vspace{0.6cm}
\begin{tikzpicture}
    \def\xbb{0.6};\def\xb{1.5};\def\xsh{1.3};\def\ysh{1};
    \def\ya{1.3};\def\yc{3.1};\def\lo{0.2};\def\loo{0.4};
    \pgfmathsetmacro\xa{-\xb};\pgfmathsetmacro\xaa{-\xbb};
    \pgfmathsetmacro\xc{\xa-\xsh/2};\pgfmathsetmacro\xd{\xa+\xsh/2};
    \pgfmathsetmacro\xe{\xb-\xsh/2};\pgfmathsetmacro\xf{\xb+\xsh/2};
    \pgfmathsetmacro\yb{\ya+\ysh};
    \draw[thick,->] (\xa,\yc) -- node[left,shift={(0.1,0.2)}] {$A_i$} (\xa,\yb);
    \draw[thick,<-] (\xb,\yc) -- node[right,shift={(-0.1,0.2)}] {$B_o$} (\xb,\yb);
    \draw (\xc,\yb) rectangle (\xd,\ya) node[midway] {\large $\cE$};
    \draw (\xe,\yb) rectangle (\xf,\ya) node[midway] {\large $\cD$};
    \draw[thick,->] (\xa,\ya) -- (\xa,0) -- node[below] {$A_o$} (\xaa,0);
    \draw[thick,<-] (\xb,\ya) -- (\xb,0) -- node[below] {$B_i$} (\xbb,0);
    \draw (\xaa,\ysh/2) rectangle (\xbb,-\ysh/2) node[midway] {\large $\cN$};
    \draw[dashed,blue,fill=blue!30,opacity=0.3] (\xc-\lo,\yb+\lo) rectangle (\xf+\lo,\ya-0.6);
    \node[blue] (dots) at (0,\yb-0.7) {\Large $\Pi$};
    \draw[thick,dotted,black!70] (\xc-\loo,\yb+\loo) rectangle (\xf+\loo,-\ysh/2-\lo);
    \node[black!50] (dots) at (\xc-0.15,-\ysh/2) {$\cM$};
\end{tikzpicture}
\caption{Bipartite operation $\Pi(A_iB_i\to A_oB_o)$ is equivalently the coding scheme ($\cE$,$\cD$) with free extra resources, such entanglement or no-signalling correlations. The whole operation is to emulate a noiseless classical (or quantum) channel $\cM(A_i\to B_o)$ using a given noisy quantum channel $\cN(A_o\to B_i)$ and the bipartite operation $\Pi$.}
\label{fig:QNSC}
\end{figure}

As described above, one can simulate a channel $\cM$ with the channel $\cN$ and code $\Pi$, where $\Pi$ is a bipartite CPTP operation from $A_iB_i$ to $A_oB_o$ which is no-signalling (NS) and PPT-preserving (PPT). In this work we shall also consider other classes of codes, such as entanglement-assisted (EA) code, unassisted (UA) code. The class of entanglement-assisted codes corresponds to bipartite operations of the form $\Pi(A_iB_i\to A_oB_o)=\cD(B_i\hat B\to B_o)\cE(A_i\hat A\to A_o)\varphi_{\hat A\hat B}$, where $\cE,\cD$ are encoding and decoding operations respectively, and $\varphi_{\hat A\hat B}$ can be any shared entangled state of arbitrary systems $\hat A$ and $\hat B$. we use $\O$ to denote specific class of codes such as $\O\in\{\text{NS},\text{PPT},\text{NS}\cap\text{PPT},\text{EA},\text{UA}\}$ in the following. 

Let $\cM(A_i\to B_o)$ denote the resulting composition channel of $\Pi$ and $\cN$, written $\cM=\Pi\circ\cN$. As both $\cM$ and $\cN$ are quantum channels, there exist quantum channels $\cE(A_i\to A_oC)$ and $\cD(B_iC\to B_o)$, where $\cE$ is an isometry operation and $C$ is a quantum register, such that \cite{Chiribella2008}
\begin{equation}
\cM(A_i\to B_o) = \cD(B_iC\to B_0) \circ \cN(A_o \to B_i) \circ \cE(A_i\to A_oC).
\end{equation}
Based on this,  the Choi-Jamio\l{}kowski matrix of $\cM$ is given by \cite{Leung2015c}
\begin{equation}
J_{\cM}=\tr_{A_oB_i} (J_{\cN}^T\ox\1_{A_iB_o})Z_{A_iA_oB_iB_o}.
\end{equation}
The operations $\cE$ and $\cD$ can be considered as generalized encoding and decoding operations respectively, except that the register $C$ may be not possessed by Alice or Bob. If the Hilbert space with $C$ is trivial, 
$\cE$ and $\cD$ become the unassisted local encoding/decoding operations. Moreover, the coding schemes $\cE,\cD$ with register $C$ can be designed to be forward-assisted codes \cite{Leung2015c}.

We are now able to derive the one-shot characterization of classical communication assisted by NS (or \text{NS}$\cap$\text{PPT}) codes.
 \begin{theorem}
 \label{theoremFidelity}
 For a given quantum channel $\cN$, the optimal success probability of $\cN$ to transmit $m$ messages assisted by NS$\cap$PPT codes is given by
\begin{equation}
\label{SDP f}
\begin{split}
 f_{{\text{\rm{NS}}\cap\text{\rm{PPT}}}} (\cN, m)= \max &\ \tr J_{\cN}F_{AB} \\
  \operatorname{s.t.}  &\  0\le F_{AB}\le \rho_A\otimes \1_B,\\
  &\ \tr\rho_A=1,\\
 &\ \tr_{A}F_{AB}=\1_B/m, \\
&\ 0\le F_{AB}^{T_B}\le \rho_A\otimes \1_B  \ (\rm{PPT}).
\end{split}\end{equation}
Similarly, when assisted by NS codes, one can remove the PPT constraint to obtain the optimal success probability as follows:
\begin{equation}
\label{SDP f NS}
\begin{split}
 f_{{\text{\rm{NS}}}} (\cN, m)= \max &\ \tr J_{\cN}F_{AB} \\
 \operatorname{s.t.}  &\ 0\le F_{AB}\le \rho_A\otimes \1_B, \\
 &\ \tr\rho_A=1,\\
 &\ \tr_{A}F_{AB}=\1_B/m.
\end{split}\end{equation}\end{theorem}
\begin{proof}
In this proof, we first use the Choi-Jamio\l{}kowski representations of quantum channels to refine the average success probability and then exploit symmetry to simplify the optimization over all possible codes. Finally, we impose the no-signalling and PPT-preserving constraints to obtain the semidefinite program of the optimal average success probability.

Without loss of generality, we assume  that $A_i$ and $B_o$ are classical registers with size $m$, i.e., the inputs and outputs are $\{\ket{k}_{A_i}\}_{k=1}^m$ and $\{\ket{k'}_{B_i}\}_{k'=1}^m$, respectively. 
For some NS$\cap$PPT code $\Pi$, the Choi-Jamio\l{}kowski  matrix of $\cM=\Pi\circ\cN$ is given by 
$J_{\cM}=\sum_{ij}\ketbra{i}{j}_{A_i}\ox\cM(\ketbra{i}{j}_{A'_i})$,
where ${A'_i}$ is isometric to $A_i$. 
Then, we can simplify $f(\cN,\Pi,m)$ to
\begin{equation}\begin{split}
f& (\cN,\Pi,m) \\
& =\frac{1}{m}\sum_{k=1}^{m}\tr\left( \cM(\proj k_{A'_i})\proj{k}_{B_o} \right) \\
&=\frac{1}{m}\tr \left( \sum_{i,j=1}^{m}(\ketbra{i}{j}_{A_i}\ox\cM(\ketbra {i}{j}_{A'_i}))\sum_{k=1}^m \ketbra{kk}{kk}_{A_iB_o} \right) \\
&=\frac{1}{m}\tr J_{\cM}\sum_{k=1}^m \ketbra{kk}{kk}_{A_iB_o}.
\end{split}\end{equation}

Then, denoting $D_{A_iB_o}=\sum_{k=1}^m \ketbra{kk}{kk}_{A_iB_o}$, we have
\begin{align*}
 f_{\text{NS}\cap\text{PPT}}(\cN, m)=\max_{\cM=\Pi\circ\cN} & \frac{1}{m}\tr (J_{\cM} D_{A_iB_o}),
\end{align*}
where $\cM=\Pi\circ\cN$ and  $\Pi$ is any feasible NS$\cap$PPT bipartite operation . (See FIG. \ref{fig:QNSC} for the implementation of $\cM$.) 
Noting  that $J_{\cM}=\tr_{A_oB_i} (J_{\cN}^T\ox\1_{A_iB_o})Z_{A_iA_oB_iB_o}$, 
we can further simplify  $f(\cN, m)$ as
\begin{equation}
\label{sdp_fidelity}
 \begin{split}
 f_{\text{NS}\cap\text{PPT}} & (\cN,m) \\
 =\max &\ \tr (J_{\cN}^T\ox \1_{A_i B_o})Z_{A_iA_oB_iB_o}(\1_{A_oB_i}\ox D_{A_iB_o})/m,\\
 \operatorname{s.t.} &\ Z_{A_iA_oB_iB_o} \text{ satisfies Eq. } (\ref{code constraints}).
 \end{split}
 \end{equation}
 
The next step is to  simplify   $f(\cN, m)$ by exploiting symmetry.  For any  permutation $\tau \in  S_m$, where $S_m$ is the symmetric group of degree $m$, if $Z_{A_iA_oB_iB_o}$ is feasible (satisfying the constraints in Eq. (\ref{code constraints})),
then it is not difficult to check that
\begin{equation}
Z'_{A_iA_oB_iB_o} = (\tau_{A_i} \ox\tau_{B_o}\ox\1_{A_oB_i}) Z_{A_iA_oB_iB_o} (\tau_{A_i} \ox\tau_{B_o}\ox\1_{A_oB_i}) ^{\dagger}
\end{equation}
is also feasible. And any convex combination $\lambda Z'+(1-\lambda) Z'' (0\le \lambda\le 1)$ of two operators satisfying Eq. (\ref{code constraints}) can also checked to be feasible. Therefore, if $Z_{A_iA_oB_iB_o}$ is feasible, so is
\begin{equation}\begin{split}
& \widetilde Z_{A_iA_oB_iB_o} = \cP_{A_iB_o} (Z_{A_iA_oB_iB_o}) \\
& := \frac{1}{m!}\sum_{\tau_{A_i},\tau_{B_o}\in S_m}(\tau_{A_i} \ox\tau_{B_o}) Z_{A_iA_oB_iB_o} (\tau_{A_i} \ox\tau_{B_o}) ^{\dagger},
\end{split}\end{equation}
where $\cP_{A_iB_o}$ is a twirling operation on $A_iB_o$.

Noticing that $\cP_{A_iB_o}(D_{A_iB_o})=D_{A_iB_o}$, we have 
\begin{equation}\begin{split}
\tr_{A_iB_o} & (Z_{A_iB_iA_oB_o}(\1_{A_oB_i}\otimes D_{A_iB_o}))  \\
=& \tr_{A_iB_o} (Z_{A_iB_iA_oB_o}(\1_{A_oB_i}\otimes \cP_{A_iB_o}(D_{A_iB_o})) \\
= &\tr_{A_iB_o} ( \widetilde Z_{A_iA_oB_iB_o}  (\1_{A_oB_i}\otimes D_{A_iB_o})).
\end{split}\end{equation}
Thus, it is easy to see that the optimal success probability equals to
\begin{align*}
 f_{\text{NS}\cap\text{PPT}} & (\cN,m) \\
 =\max &\ \tr (J_{\cN}^T\ox \1_{A_i B_o})\widetilde Z_{A_iA_oB_iB_o}(\1_{A_oB_i}\ox D_{A_i B_o})/m \\
  \operatorname{s.t.} &\  \widetilde Z_{A_iA_oB_iB_o} \text{ satisfies Eq. } (\ref{code constraints}).
  \end{align*}
  
It is worth noting that $\widetilde Z_{A_iA_oB_iB_o}$ can be rewritten as \cite{Duan2016}
$$\widetilde Z_{A_iA_oB_iB_o}=F_{A_oB_i}\otimes D_{A_iB_o} +E_{A_oB_i}\otimes (\1-D_{A_iB_o}),$$
for some operators $E_{A_oB_i}$ and $F_{A_oB_i}$.
Thus, the objective function can be simplified to
$\tr J_{\cN}^TF$.
Also, the  CP and PPT constraints are equivalent to \begin{equation}\label{F PSD}
E_{A_oB_i}\ge 0,F_{A_oB_i}\ge 0, E_{A_oB_i}^{T_{B_i}}\ge 0,F_{A_oB_i}^{T_{B_i}}\ge 0.
\end{equation}
Furthermore,
the $B \not\rightarrow A$ constraint is equivalent to $\tr_{B_o}\widetilde Z_{A_iA_oB_iB_o}=\tr_{B_oB_i}\widetilde Z_{A_iA_oB_iB_o}\otimes {\1_{B_i}}/{d_{B_i}}$, i.e.
\begin{equation}\label{E F}
\begin{split}
& F _{A_oB_i}+(m-1)E_{A_oB_i} \\
& = \tr_{B_i}(F_{A_oB_i}+(m-1)E_{A_oB_i})\otimes\frac{\1_{B_i}}{d_{B_i}}=:\rho_{A_o}\otimes\1_{B_i}.
\end{split}\end{equation}
and the TP constraint holds if and only if $\tr_{A_oB_o}Z_{A_iA_oB_iB_o}=\1_{A_iB_i}$, i.e., 
\begin{equation}
\tr_{A_o}(F_{A_oB_i}+(m-1)E_{A_oB_i})=\1_{B_i},
\end{equation}
 which is equivalent to 
 \begin{equation}\label{rho state}
 \tr\rho_{A_o}=\tr(F_{A_oB_i}+(m-1)E_{A_oB_i})/d_{B_i}=\tr \1_{B_i}/d_{B_i}=1.
 \end{equation}

As $\Pi$ is no-signalling from A to B, we have $\tr_{A_o}\widetilde Z_{A_iA_oB_iB_o}=\tr_{A_oA_i}\widetilde Z_{A_iA_oB_iB_o}\otimes \frac{\1_{A_i}}{m}$, i.e.,
\begin{equation}\begin{split}
&\tr_{A_o}F_{A_oB_i}\otimes D_{A_iB_o}+\tr_{A_o}E_{A_oB_i}\otimes (\1-D_{A_iB_o}) \\
=&\tr_{A_o}(F_{A_oB_i}+(m-1)E_{A_oB_i})\otimes\frac{\1_{A_iB_o}}{m}=\1_{A_iB_iB_o}/m.
\end{split}\end{equation}
Since $D_{A_iB_o}$ and $\1-D_{A_iB_o}$ are orthogonal positive operators, 
we have 
\begin{equation}\label{F NS}
\tr_{A_o}F_{A_oB_i}=\tr_{A_o}E_{A_oB_i}=\1_{B_i}/m.
\end{equation} 

Finally, combining Eq. (\ref{F PSD}),  (\ref{E F}), (\ref{rho state}), (\ref{F NS}), we have that
\begin{equation}\begin{split}
 f_{\text{NS}\cap\text{PPT}} (\cN, m) = \max&\ \tr J_{\cN}F_{A_oB_i} \\
  \rm{ s.t. }  &\  0\le F_{A_oB_i}\le \rho_{A_o}\otimes \1_{B_i},\\
  &\ \tr\rho_{A_o}=1,\\
 &\ \tr_{A_o}F_{A_oB_i}=\1_{B_i}/m,\\
&\ 0\le F_{A_oB_i}^{T_{B_i} }\le \rho_{A_o}\otimes \1_{B_i} .
\end{split}\end{equation}
This gives the SDP in Theorem \ref{theoremFidelity}, where
we assume that $A_o=A$ and $B_i=B$ for simplification.
\end{proof}
\textbf{Remark:}
The dual SDP for $f_{\text{NS}\cap\text{PPT}}(\cN, m)$ is given by
\begin{equation}\label{dual SDP f} \begin{split}
 f_{\text{NS}\cap\text{PPT}} (\cN, m) = \min & \ t+\tr S_B/m  \\
 \text{s.t.}  &  J_{\cN}\le X_{AB}+\1_A\ox S_B+(W_{AB}-Y_{AB})^{T_B}, \\
 &\tr_B(X_{AB}+W_{AB})\le t\1_A,\\
 & X_{AB},Y_{AB},W_{AB}\ge0, S_B=S_B^\dagger.\\
\end{split}\end{equation}
To remove the PPT constraint, set $Y_{AB}=W_{AB}=0$.
It is worth noting that the strong duality holds here  since the Slater's condition can be easily checked. Indeed, choosing $X_{AB}=Y_{AB}=W_{AB}=\|J_{\cN}\|_{\infty}\1_{AB}$, $S_B=\1_B$ and $t=3 d_B \|J_{\cN}\|_{\infty}$ in SDP (\ref{dual SDP f}), we have $(X_{AB},Y_{AB},W_{AB},S_B,t)$ is in the relative interior of the feasible region. 

It is worthing noting that $f_{\text{NS}} (\cN, m)$ can be obtained by removing the PPT constraint and it corresponds with the optimal NS-assisted channel fidelity in \cite{Leung2015c}.


\subsection{Improved SDP converse bounds in finite blocklength}
For  given $0\le \e <1$, the {\em one-shot $\e$-error classical capacity assisted by $\O$-class codes} is defined as
\begin{equation}\label{error capacity}
C_\O^{(1)}(\cN,\e):=\sup\{\log\lambda: 1-f_\O(\cN,\lambda)\le \e\}.
\end{equation}

We now derive the one-shot $\e$-error classical capacity assisted by NS or NS$\cap$PPT codes as follows.

\begin{theorem}\label{converse PPT}
For given channel $\cN$ and error threshold $\e$, the one-shot $\e$-error \rm{NS$\cap$PPT}-assisted and \rm{NS}-assisted  capacities are given by
\begin{equation}\label{C1 NSPPT}
\begin{split}
C_{\text{\rm{NS}}\cap\text{\rm{PPT}}}^{(1)} (\cN,\e)= -\log \min &\  \eta \\ 
\operatorname{s.t.}  &\  0\le F_{AB}\le \rho_A\otimes \1_B, \\
&\ \tr\rho_A=1, \tr_{A}F_{AB}=\eta\1_B, \\
&\ \tr J_{\cN}F_{AB}\ge1-\e, \\
&\  0\le F_{AB}^{T_B}\le \rho_A\otimes \1_B \ (\rm{PPT}), \\
\end{split}\end{equation}
and
\begin{equation}\label{C1 NS}
\begin{split}
C_{\rm{NS}}^{(1)} (\cN,\e) =-\log \min &\ \eta  \\ \rm{s.t.} &\  0\le F_{AB}\le \rho_A\otimes \1_B,\tr\rho_A=1, \\
&\tr_{A}F_{AB}=\eta\1_B, 
\tr J_{\cN}F_{AB}\ge1-\e,\\
\end{split}\end{equation}
respectively.
\end{theorem}
\begin{proof}
When assisted by NS$\cap$PPT codes, by Eq. (\ref{error capacity}), we have that
\begin{equation}\label{R error}
\begin{split}
C_{\text{NS}\cap\text{PPT}}^{(1)}(\cN,\e)=\log \max \lambda \ \rm{ s.t. } f_{\text{NS}\cap\text{PPT}}(\cN,\lambda)\ge 1-\e.
\end{split}\end{equation}
To simplify Eq. (\ref{R error}), we suppose that
\begin{equation}\label{tilde R}
\begin{split}
\U(\cN,\e)= -\log \min &\ \eta \\
\operatorname{s.t.}  &\  0\le F_{AB}\le \rho_A\otimes \1_B, \\
&\ \tr\rho_A=1, \tr_{A}F_{AB}=\eta\1_B, \\
 &\ \tr J_{\cN}F_{AB}\ge 1-\e, \\
&\ 0\le F_{AB}^{T_B}\le \rho_A\otimes \1_B \ (\rm{PPT}).\\
\end{split}
\end{equation}

On one hand, for given $\e$, suppose that the optimal solution to the SDP (\ref{tilde R}) of $\U(\cN,\e)$ is $\{\rho,F,\eta\}$. Then, it is clear that $\{\rho,F\}$ is a feasible solution of the SDP (\ref{SDP f}) of
$f_{\text{NS}\cap\text{PPT}}(\cN,\eta^{-1})$, which means that $f_{\text{NS}\cap\text{PPT}}(\cN,\eta^{-1})\ge \tr J_\cN F\ge 1-\e$. Therefore, 
\begin{equation}\label{R 1}
C^{(1)}_{\text{NS}\cap\text{PPT}}(\cN,\e)\ge \log \eta^{-1}=\U(\cN,\e).
\end{equation}

On the other hand, for given $\e$, suppose that the value of $C^{(1)}_{\text{NS}\cap\text{PPT}}(\cN,\e)$ is $\log\lambda$ and the optimal solution of $f_{\text{NS}\cap\text{PPT}}(\cN,\lambda)$ is $\{\rho,F\}$. It is  easy to check that $\{\rho,F,\lambda^{-1}\}$ satisfies the constrains in SDP (\ref{tilde R}) of $\U(\cN,\e)$. Therefore, 
\begin{equation}\label{R 2}
\U(\cN,\e)\ge -\log \lambda^{-1}= C^{(1)}_{\text{NS}\cap\text{PPT}}(\cN,\e).\end{equation}

Hence, combining Eqs. (\ref{tilde R}), (\ref{R 1}) and (\ref{R 2}), it is clear that
\begin{equation}\begin{split}
C_{\text{\rm{NS}}\cap\text{\rm{PPT}}}^{(1)}(\cN,\e)=\U(\cN,\e)\\
= -\log \min &\ \eta \\  \rm{ s.t. }  &\  0\le F_{AB}\le \rho_A\otimes \1_B, \\ 
&\ \tr\rho_A=1, \tr_{A}F_{AB}=\eta\1_B, \\
&\ \tr J_{\cN}F_{AB}\ge1-\e, \\
&\ 0\le F_{AB}^{T_B}\le \rho_A\otimes \1_B \ (\rm{PPT}).\\
\end{split}\end{equation}
And one can obtain $C_{\text{\rm{NS}}}^{(1)}(\cN,\e)$ by removing the  PPT constraint.
\end{proof}

Noticing that no-signalling-assisted codes are potentially stronger than the entanglement-assisted codes, $C_{\text{\rm{NS}}}^{(1)}(\cN,\e)$ and 
$C_{\text{\rm{NS}}\cap\text{\rm{PPT}}}^{(1)}(\cN,\e)$ provide converse bounds of classical communication for entanglement-assisted and unassisted codes, respectively.
\begin{corollary}
For a given channel $\cN$ and error threshold $\e$, 
\begin{align*}
C^{(1)}_{\rm{E}}(\cN,\e)&\le C^{(1)}_{\rm{NS}}(\cN,\e),\\
C^{(1)}(\cN,\e)&\le C^{(1)}_{\rm{NS}\cap {PPT}}(\cN,\e).
\end{align*}
\end{corollary}

We further compare our one-shot $\e$-error capacities with the previous SDP converse bounds derived by the quantum hypothesis testing technique in   \cite{Matthews2014}.
To be specific, for a given channel $\cN(A\to B)$ and error thresold $\e$,   Matthews and Wehner   \cite{Matthews2014}  establish that
\begin{equation}\label{R E}
\begin{split}
C^{(1)}_{\rm{E}}& (\cN,\e) \le R_{\rm{E}}(\cN,\e)\\
&=\max_{\rho_A}\min_{\sigma_B}D_H^\e((id_{A'}\ox\cN)(\rho_{A'A})||\rho_{A'}\ox \sigma_B)\\
&=-\log \min \eta   \\
&\phantom{==-\log}  \rm{ s.t. }    0\le F_{AB}\le \rho_A\otimes \1_B,\tr\rho_A=1,\\
&\phantom{ ==-\log \min } \tr_{A}F_{AB}\le \eta\1_B, \tr J_{\cN}F_{AB}\ge 1-\e,
\end{split}
\end{equation}
and 
\begin{equation}
\begin{split}
C^{(1)}  (\cN,\e)\le R & _{\rm{E\cap PPT}} (\cN,\e)\\
=\max_{\rho_A}\min_{\sigma_B} & D_{H,PPT}^\e((id_{A'}\ox\cN)(\rho_{A'A})||\rho_{A'}\ox \sigma_B)\\
=-\log \min &\ \eta  \\
\rm{ s.t. } &\   0\le F_{AB}\le \rho_A\otimes \1_B,\tr\rho_A=1, \\
&\     \tr_{A} F_{AB}\le \eta\1_B, \tr J_{\cN}F_{AB}\ge 1-\e,\\
 &\     0\le F_{AB}^{T_B}\le \rho_A\otimes \1_B,
\end{split}
\end{equation}
where $\rho_{A'A}=(\1_{A'}\ox\rho_A^{\frac{1}{2}})\Phi_{A'A}(\1_{A'}\ox\rho_A^{\frac{1}{2}})$ is a purification of $\rho_A$ and $\rho_{A'}=\tr_A\rho_{A'A}$.
Moreover, 
\begin{equation}
\begin{split}
D_H^\e(\rho_0||\rho_1)=-\log\min&\ \tr T\rho_1 \\
\rm{ s.t. }&\ 1-\tr T\rho_0\le\e, 0\le T \le \1
\end{split}
\end{equation}
 is the hypothesis testing relative entropy \cite{Wang2012,Matthews2014} and $D_{H,PPT}^\e(\rho_0||\rho_1)$ is the similar quantity with a PPT constraint on the POVM.

Interestingly, our one-shot $\e$-error capacities are similar to these quantum hypothesis testing relative entropy converse bounds.
However, there is a crucial difference that our quantities require that a stricter condition, i.e., $\tr_{A}F_{AB}= \eta\1_B$. This makes
one-shot $\e$-error capacities ($C_{\text{\rm{NS}}\cap\text{\rm{PPT}}}^{(1)}(\cN,\e)$ and $C_{\text{\rm{NS}}}^{(1)}(\cN,\e)$) always smaller than or equal to 
the SDP converse bounds in \cite{Matthews2014}, and the inequalities can be strict.
\begin{proposition}\label{R NS R E}
For a given channel $\cN(A\to B)$ and error threshold $\e$,
\begin{align*}
C^{(1)}_{\rm{NS}} & (\cN,\e)\le R_{\rm{E}}(\cN,\e) \\
=& \max_{\rho_A}\min_{\sigma_B}D_H^\e((id_{A'}\ox\cN)(\rho_{A'A})||\rho_{A'}\ox \sigma_B),\\
 C^{(1)}_{\rm{NS}\cap {PPT}} & (\cN,\e) \le R_{\rm{E\cap PPT}}(\cN,\e) \\
=& \max_{\rho_A}\min_{\sigma_B}D_{H,PPT}^\e((id_{A'}\ox\cN)(\rho_{A'A})||\rho_{A'}\ox \sigma_B).
\end{align*}
In particular, both inequalities can be strict for some quantum channels such as the amplitude damping channels and  the simplest classical-quantum channels.
\end{proposition}
\begin{proof}
This can be proved by the fact that any feasible solution of the SDP (\ref{C1 NS}) of 
$C^{(1)}_{\text{NS}}(\cN,\e)$ (or $C_{\text{NS}\cap\text{PPT}}^{(1)}(\cN,\e)$) is also feasible to the SDP (\ref{R E}) of $R_{\rm{E}}(\cN,\e)$ (or $R_{\rm{E\cap PPT}}(\cN,\e)$). 

We further show that the inequality can be strict by the example of qubit amplitude damping channel  $\cN_{\gamma}^{AD}=\sum_{i=0}^1 E_i\cdot E_i^\dag$ $(0\le \gamma\leq 1)$,
with $E_0=\ketbra{0}{0}+\sqrt{1-\gamma}\ketbra{1}{1}$ and $E_1=\sqrt{\gamma}\ketbra{0}{1}$. We compare the above bounds in FIG. \ref{ad1} and FIG. \ref{ad2}. It is clear that our bounds can be strictly better than the quantum hypothesis testing bounds in \cite{Matthews2014} in this case.
\begin{figure}[htbp]
\vspace{-0.15cm}
\centering
\begin{minipage}{.48\textwidth}
 \centering{\includegraphics[width=5.8cm]{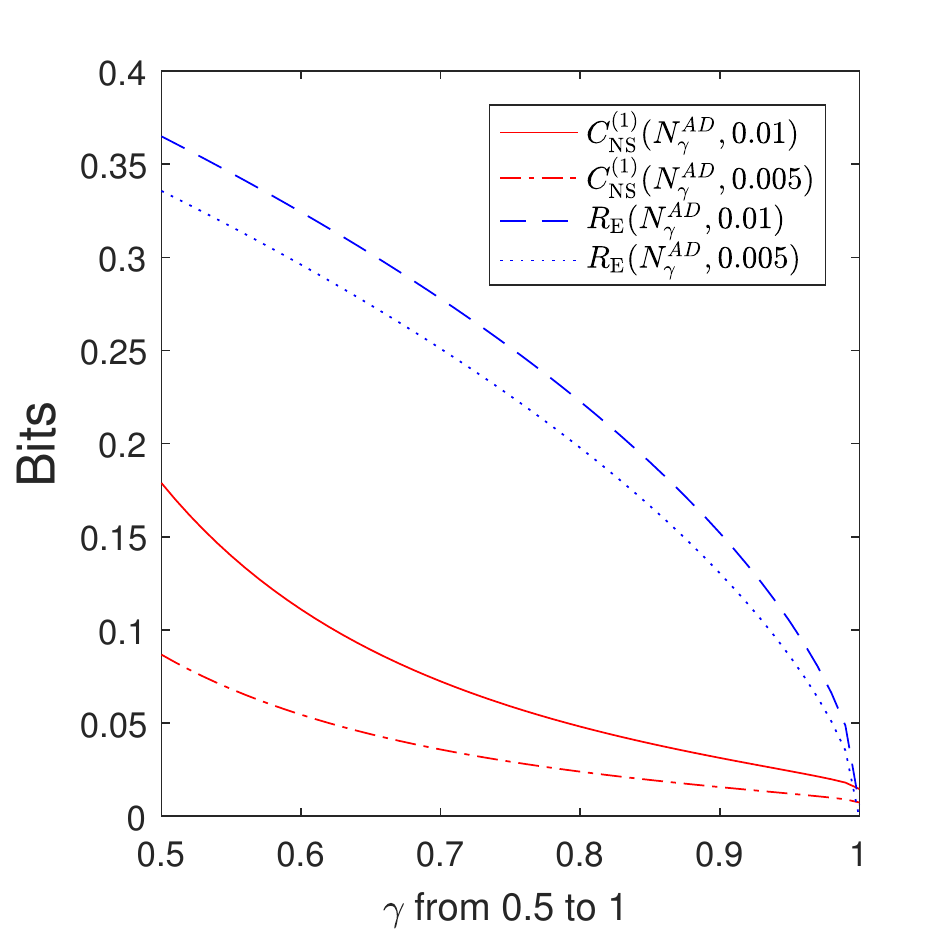}}
  \caption{ The red solid and dash-dot lines depict $C_{\rm{NS}}^{(1)}(\cN_{\gamma}^{AD},0.01)$  and $C_{\rm{NS}}^{(1)}(\cN_{\gamma}^{AD},0.005)$, respectively.
 The blue dashed and dotted lines depict  $R_{\rm{E}}(\cN_{\gamma}^{AD},0.01)$ and  $R_{\rm{E}}(\cN_{\gamma}^{AD},0.005)$.}
\label{ad1}
\end{minipage}
\vspace{0.3cm}
\begin{minipage}{.48\textwidth}
 \centering{\includegraphics[width=5.6cm]{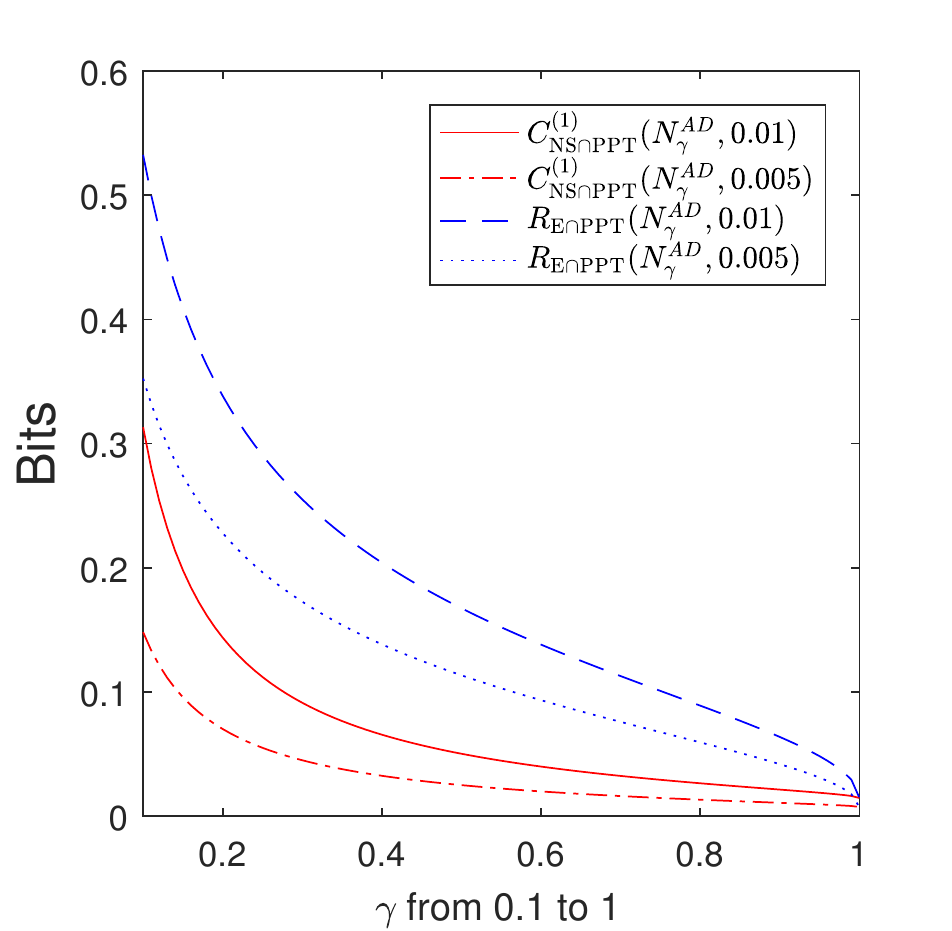}}
  \caption{ The red solid and dash-dot lines depict $C_{\rm{NS\cap PPT}}^{(1)}(\cN_{\gamma}^{AD},0.01)$  and $C_{\rm{NS\cap PPT}}^{(1)}(\cN_{\gamma}^{AD},0.005)$, respectively.
 The blue dashed and dotted lines depict  $R_{\rm{E\cap PPT}}(\cN_{\gamma}^{AD},0.01)$ and  $R_{\rm{E\cap PPT}}(\cN_{\gamma}^{AD},0.005)$, respectively. }
\label{ad2}
\end{minipage}
\vspace{-0.5cm}
\end{figure}

Another example is the simplest classical-quantum channel $\cN_{a}^{cq}$ which has only two inputs and two pure output states $ \proj{\psi_i}$,
w.l.o.g.
\begin{align*}
  \ket{\psi_0} = a\ket{0} + b\ket{1},  \\
  \ket{\psi_1} = a\ket{0} - b\ket{1},
\end{align*}
with $a \geq b = \sqrt{1-a^2}$.
The comparison is presented in FIG. \ref{cq1} and it is clear that our bound can be strictly tighter for this class of classical-quantum channels.

\end{proof}
\begin{figure}[htbp]
 \centering{\includegraphics[width=6.0cm]{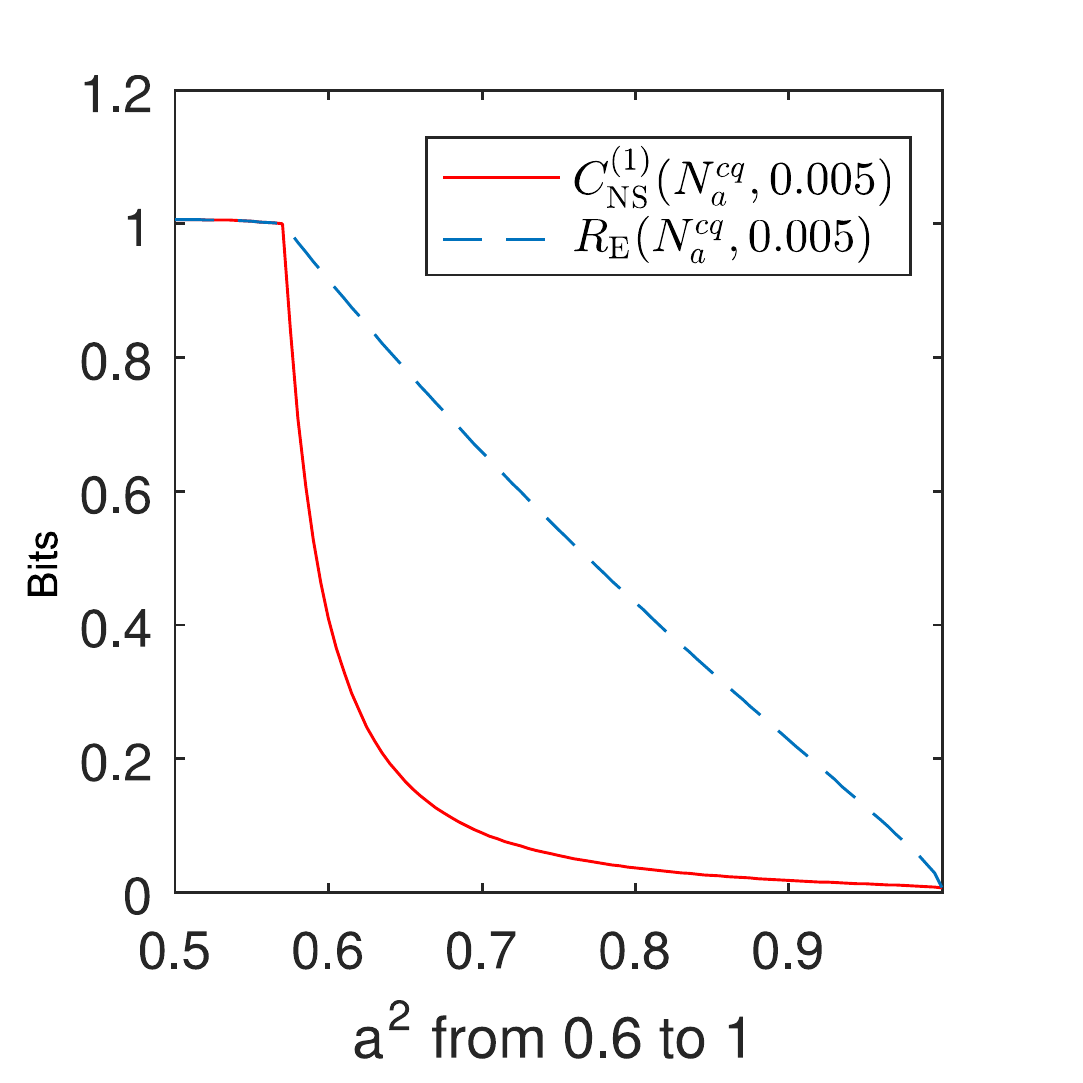}}
  \caption{When $\e=0.005$,  $C_{\rm{NS}}^{(1)}(\cN_{a}^{cq},\e)$ (red solid line) can be strictly smaller than the previous SDP bound  $R_{\rm{E}}(\cN_{a}^{cq},\e)$ (blue dashed line). Note that $C_{\rm{NS\cap PPT}}^{(1)}(\cN_{a}^{cq},\e)=C_{\rm{NS}}^{(1)}(\cN_{a}^{cq},\e)$ and  $R_{\rm{E\cap PPT}}(\cN_{a}^{cq},\e)=R_{\rm{E}}(\cN_{a}^{cq},\e)$ in this case.}
\label{cq1}
\end{figure}

We then consider the asymptotic performance of the one-shot NS-assisted $\e$-error capacity.
Interestingly, in common with the bound $R_{\rm E}(\cN,\e)$ \cite{Matthews2014} and the bound of Datta and Hsieh \cite{Datta2013c}, the asymptotic behaviour of $R_{\text{NS}}(\cN,\e)$ also recovers the converse part of the formula for entanglement-assisted capacity \cite{Bennett1999} and it implies that $C_{\rm{NS}}(\cN)=C_{\rm E}(\cN)$. (See Corollary \ref{C E C NS}.) In \cite{Leung2015c}, Leung and  Matthews  have already shown that the entanglement-assisted quantum capacity of a quantum channel is equal to the NS-assisted quantum capacity. 
It is worth noting  that our result is equivalent to their result due to superdense coding \cite{Bennett1992} and teleportation \cite{Bennett1993}.
\begin{corollary}\label{C E C NS}
For any quantum channel $\cN (A\to B)$,
\begin{equation*}
\lim_{\e\to0}\lim_{n\to \infty}\frac{1}{n} C_{\text{\rm{NS}}}^{(1)}(\cN^{\ox n},\e)\le \max_{\rho_A} I(\rho_A;\cN),
\end{equation*}
where $I(\rho_A;\cN):=H(\rho_A)+H(\cN(\rho_A))-H((\rm{id}\ox\cN)\phi_{\rho_A})$, and $\phi_{\rho_A}$ is a purification of $\rho_A$.
As a consequence, 
\begin{equation*}
C_{\rm{NS}}(\cN)=C_{\rm E}(\cN).
\end{equation*}
\end{corollary}
\begin{proof}
In \cite{Matthews2014}, Matthews and Wehner prove that 
$$\lim_{\e\to0}\lim_{n\to \infty}\frac{1}{n} R_{\rm{E}}(\cN^{\ox n},\e)\le  \max_{\rho_A} I(\rho_A;\cN).$$

By Proposition \ref{R NS R E}, we immediately obtain that
\begin{equation}
\begin{split}
\lim_{\e\to0}\lim_{n\to \infty}\frac{1}{n} C^{(1)}_{\rm{NS}}(\cN^{\ox n},\e) &\le \lim_{\e\to0}\lim_{n\to \infty}\frac{1}{n} R_{\rm{E}}(\cN^{\ox n},\e) \\
&\le  \max_{\rho_A} I(\rho_A;\cN),
\end{split}
\end{equation}
which means that $C_{\rm{NS}}(\cN) \le C_{\rm{E}}(\cN)$.
Noticing that no-signalling codes are potentially stronger than the entanglement codes, it holds  that $C_{\rm{NS}}(\cN) \ge C_{\rm E}(\cN)$. Therefore, we have that $C_{\rm{NS}}(\cN)=C_{\rm {E}}(\cN)$.
\end{proof}

\subsection{Reduction to Polyanskiy-Poor-Verd\'u converse bound}
For classical-quantum channels, the one-shot $\e$-error NS-assisted (or NS$\cap$PPT-assisted)
capacity can be further simplified based on the structure of the channel. 
\begin{proposition}
For the classical-quantum channel that acts as $\cN: x\to \rho_x$, the Choi matrix of $\cN$ is given by $J_\cN=\sum_x \proj x\ox \rho_x$. Then, the SDP (\ref{C1 NS}) of $C_{\rm{NS}}^{(1)}(\cN,\e)$ and the SDP (\ref{C1 NSPPT}) of $C_{\rm{NS\cap PPT}}^{(1)}(\cN,\e)$ 
can be simplified to
\begin{equation}\label{C cq SDP}
\begin{split}
C_{\rm{NS}}^{(1)}(\cN,\e)=C_{\rm{NS\cap PPT}}^{(1)} & (\cN,\e) \\
=\log \max &\ \sum{s_x}   \\ \rm{ s.t. } &\    0\le Q_x\le s_x \1_B, \forall x,\\
&\sum_x Q_x= \1_B, \\
&\sum_x\tr Q_x\rho_x\ge \sum_x(1-\e)s_x.
\end{split}\end{equation}

\end{proposition}
\begin{proof}
When $J_\cN=\sum_x \proj x\ox \rho_x$,  the SDP (\ref{C1 NS}) easily simplifies to
\begin{equation}\label{C cq SDP 1}
\begin{split}
C_{\rm{NS}}^{(1)}(\cN,\e)=-\log \min &\ \eta \\ \rm{ s.t. } &\    0\le F_x\le p_x \1_B, \forall x,\\
&\sum_x p_x=1,\\
&\sum_x F_x/\eta= \1_B, \\
&\sum_x\tr F_x\rho_x\ge (1-\e).
\end{split}\end{equation}

By assuming that $Q_x=F_x/\eta$ and $s_x=p_x/\eta$, the above SDP simplifies to
\begin{equation}\label{C cq SDP 2}
\begin{split}
C_{\rm{NS}}^{(1)}(\cN,\e)=\log \max &\ \sum s_x \\ \rm{ s.t. } &\    0\le Q_x\le s_x \1_B, \forall x,\\
&\sum_x Q_x= \1_B, \\
&\sum_x\tr Q_x\rho_x\ge (1-\e)\sum s_x,
\end{split}\end{equation}
where we use the fact $\sum s_x=\sum p_x/\eta=1/\eta$. One can use a similar method to simplify $C_{\rm{NS\cap PPT}}^{(1)}(\cN,\e)$ as well.
\end{proof}

Furthermore, for the classical channels, Polyanskiy, Poor, and Verd\'u \cite{Polyanskiy2010} derive the finite blocklength converse via hypothesis testing.   In \cite{Matthews2012}, an alternative proof of PPV converse was provided by considering the assistance of the classical no-signalling correlations. Here, we are going to show that 
both $C_{\rm{NS}}^{(1)}(\cN,\e)$ and $C_{\rm{NS\cap PPT}}^{(1)}(\cN,\e)$ will   reduce to the PPV converse.

Let us first recall the linear program for the PPV converse bound of a classical channel $\cN(y|x)$ \cite{Polyanskiy2010,Matthews2012}:
\begin{equation}
\label{LP PPV}
    \begin{split}
    R^{PPV}(\cN,\e)=\max &\ \sum_{x} s_x \\ \rm{ s.t. } &\ Q_{xy}\le s_x, \forall x,y,\\
    &\sum_x Q_{xy}\le 1, \forall y,\\
    &\sum_{x,y} \cN(y|x)Q_{xy}\ge (1-\e)\sum_x s_x.
    \end{split}
\end{equation}

For classical channels, we can further simplify the SDP (\ref{C cq SDP}) to a linear program which coincides with the Polyanskiy-Poor-Verd\'u converse bound.

\begin{proposition}
For a classical channel $\cN(y|x)$,
\begin{equation}
    C_{\rm{NS}}^{(1)}(\cN,\e)=C_{\rm{NS\cap PPT}}^{(1)}(\cN,\e)=R^{PPV}(\cN,\e).
\end{equation}
\end{proposition}
\begin{proof}
The idea is to further simplify the SDP (\ref{C cq SDP}) via the structure of classical channels. For input $x$, the corresponding outputs can be seemed as $\rho_x=\sum_y \cN(y|x)\proj{y}$. Then, $Q_x$ should be 
diagonal for any $x$, i.e., $Q_x=\sum_y Q_{xy}$.
Thus,
 SDP (\ref{C cq SDP}) can be easily simplified to
\begin{equation}
\label{c classical channel}
    \begin{split}
    C_{\rm{NS}}^{(1)}(\cN,\e)= C_{\rm{NS\cap PPT}}^{(1)} & (\cN,\e)\\
    =    \log\max &\ \sum_{x}s_x \\ \rm{ s.t. } &\ Q_{xy}\le s_x, \forall x,y,\\
    &\sum_x Q_{xy}= 1, \forall y,\\
    &\sum_{x,y}  \cN(y|x)Q_{xy}\ge (1-\e)\sum_x s_x.
    \end{split}
\end{equation}
Using the similar technique in \cite{Matthews2012}, the constraint $\sum_x Q_{xy}= 1$ can be relaxed to $\sum_x Q_{xy}\le 1$ in this case, which means that the linear program (\ref{c classical channel}) is equal to the  linear program (\ref{LP PPV}).
\end{proof}

\section{Strong converse bounds for classical communication}
\label{NS PPT asymptotics}
\subsection{SDP strong converse bounds for the classical capacity}
It is well known that evaluating the classical capacity of a general channel is extremely difficult. To the best of our knowledge, the only known nontrivial strong converse bound for the classical capacity is the entanglement-assisted capacity \cite{Bennett1999} and there is also computable single-shot upper bound derived from entanglement measures \cite{Brandao2011c}.
In this section, we will  derive two  SDP strong converse bounds for the classical capacity of a general quantum channel. Our bounds are efficiently computable and do not depend on any special properties of the channel. We also show that for some classes of quantum channels,  our bound can be strictly smaller than the entanglement-assisted capacity and the previous bound in \cite{Brandao2011c}.

Before introducing the strong converse bounds, we first show a single-shot SDP to estimate the optimal success probability of classical communication via multiple uses of the channel.
\begin{proposition}\label{F A UB}
For any quantum channel $\cN$ and given $m$,
$$f_{\text{\rm{NS}}\cap\text{\rm{PPT}}}(\cN, m)\le f^{+}(\cN,m),$$
where
\begin{equation}\label{prime F+}
\begin{split}
 f^{+}(\cN,m)=\min &\ \tr Z_B \\  \rm{ s.t. } &\ -R_{AB}\le J_{\cN}^{T_B}\le R_{AB},  \\
 & -m\1_A\ox Z_B\le R_{AB}^{T_B}\le m\1_A\ox Z_B.
\end{split}\end{equation}

Furthermore, it holds that
$f_{\text{\rm{NS}}\cap\text{\rm{PPT}}}(\cN_1\ox\cN_2, m_1m_2)\le f^{+}(\cN_1, m_1)f^{+}(\cN_2,m_2)$.
Consequently,
\begin{equation}
 f_{\text{\rm{NS}}\cap\text{\rm{PPT}}}(\cN^{\ox n}, m^n)\le f^{+}(\cN, m)^n.
\end{equation}
\end{proposition}

\begin{proof}
We utilize the duality theory of semidefinite programming in the proof. To be specific,
the dual SDP of $f^{+}(\cN,m)$ is given by
\begin{equation}\label{dual f+}
\begin{split}
f^{+}(\cN,m)= \max &\ \tr J_{\cN}(V_{AB}-X_{AB} )^{T_B}\\
   \rm{ s.t. }&\  V_{AB}+X_{AB}  \le (W_{AB}-Y_{AB})^{T_B}, \\
   &\tr_A (W_{AB}+Y_{AB} )\le  \1_B/m, \\
   & V_{AB},X_{AB},W_{AB},Y_{AB} \ge 0.
\end{split}\end{equation}
It is worth noting that the optimal values of the primal and the dual SDPs above coincide. This is a consequence of strong duality. By Slater's condition, one simply needs to show that there exists positive definite $V_{AB}$, $X_{AB}$, $W_{AB}$ and $Y_{AB}$ such that $V_{AB}+X_{AB}  <(W_{AB}-Y_{AB})^{T_B}$ and $\tr_A (W_{AB}+Y_{AB} )<  \1_B/m$, which holds for $W_{AB}=2Y_{AB}=5V_{AB}=X_{AB}=\1_{AB}/2md_{A}$. 

In SDP (\ref{dual f+}), let us choose $X_{AB}=Y_{AB}=0$ and $V_{AB}^{T_B}=W_{AB}$, then we have that
\begin{equation} \label{f+ lower}
\begin{split}
f^{+}(\cN,m) \ge K    \ge f_{\text{NS}\cap\text{PPT}}(\cN,m),    
\end{split}\end{equation}
where $K:=\max\{\tr J_{\cN}W_{AB}:  W_{AB}, W_{AB}^{T_B}\ge 0,  \tr_A W_{AB} \le  \1_B/m \}$. This means that the SDP (\ref{dual f+}) of $f^+(\cN,m)$   is a relaxation of the SDP  (\ref{SDP f}) of $f_{\text{NS}\cap\text{PPT}}(\cN,m)$.

To see  $f_{\text{NS}\cap\text{PPT}}(\cN_1\ox\cN_2, m_1m_2)\le f^{+}(\cN_1,m_1)f^{+}(\cN_2,m_2)$, we first suppose that
the optimal solution to SDP (\ref{prime F+}) of $f^+(\cN_1, m_1)$ is $ \{Z_1,R_1\}$  and
the optimal solution to SDP (\ref{prime F+}) of
$ f^+(\cN_2, m_2)$ is $\{Z_2,R_2\}$. Let us denote the  Choi-Jamio\l{}kowski  matrix of $\cN_1$ and $\cN_2$ by $J_1$ and $J_2$, respectively.
It is easy to see that $R_1  \ox R_2\pm J_1^{T_B}\ox J_2^{T_{B'}}\ge 0$ since
\begin{align*}
& R_1  \ox R_2+J_1^{T_B}\ox J_2^{T_{B'}} \\
 =& \frac{1}{2}[(R_1+J_1^{T_B})\ox (R_2+J_2^{T_{B'}})
 +(R_1-J_1^{T_B})\ox (R_2-J_2^{T_{B'}})],\\
& R_1  \ox R_2-J_1^{T_B}\ox J_2^{T_{B'}} \\
 =& \frac{1}{2}[(R_1+J_1^{T_B})\ox (R_2-J_2^{T_{B'}})
 +(R_1-J_1^{T_B})\ox (R_2+J_2^{T_{B'}})].\\
\end{align*}
Therefore,  we have that $$-R_1\ox R_2 \le J_1^{T_B}\ox J_2^{T_{B'}}\le R_1\ox R_2.$$
Applying similar techniques, it is easy to prove that 
 $$  -m_1m_2\1_{AA'}\ox Z_1\ox Z_2 \le  R_1^{T_B}\ox R_2^{T_{B'}} \le m_1m_2\1_{AA'}\ox Z_1\ox Z_2.$$

Hence, $\{Z_1\ox Z_2, R_1\ox R_2\}$ is a feasible solution to the SDP (\ref{prime F+}) of $f^{+}(\cN_1\ox\cN_2, m_1m_2)$, which means that 
\begin{align*}
f_{\text{NS}\cap\text{PPT}}& (\cN_1\ox\cN_2, m_1m_2)\le f^{+}(\cN_1\ox\cN_2, m_1m_2) \\
&\le \tr Z_1\ox Z_2= f^{+}(\cN_1,m_1)f^{+}(\cN_2,m_2).
\end{align*}
\end{proof}

Now, we are able to derive the strong converse bounds of the classical capacity.
\begin{theorem}\label{SDP strong converse}
For any quantum channel $\cN$,
\begin{align*}
C  (\cN)\le C_{\rm{NS\cap PPT}}(\cN)
\le C_{\beta}(\cN) 
= \log  \beta(\cN)\le \log (d_{B}\|J_{\cN}^{T_B}\|_\infty),
\end{align*}
where
\begin{equation}
\label{sdp_beta}
\begin{split}
\beta(\cN)=\min & \tr S_B \\
 \rm{ s.t. } &  -R_{AB}\le J_{\cN}^{T_B}\le R_{AB}, \\
 &  -\1_A\ox S_B\le R_{AB}^{T_B}\le \1_A\ox S_B.
\end{split}\end{equation}

In particular, when the communication rate exceeds $C_{ \beta}(\cN)$, the error probability goes to one exponentially fast as the number of channel uses increases.
\end{theorem}
\begin{proof}
For $n$ uses of the channel, we suppose that the rate of the communication is $r$.
By Proposition \ref{F A UB}, we have that
\begin{equation}
\begin{split}
        f_{\text{NS}\cap\text{PPT}}(\cN^{\ox n}, 2^{rn})\le
          f^{+}(\cN, 2^r)^n.
\end{split}\end{equation}
Therefore, the $n$-shot error probability satisfies that
\begin{equation}\label{n-shot error}
\e_n=1-f_{\text{NS}\cap\text{PPT}}(\cN^{\ox n}, 2^{rn})\ge 1-f^{+}(\cN, 2^r)^n.
\end{equation}

Suppose that the optimal solution to the SDP (\ref{sdp_beta}) of $\b(\cN)$ is $\{S_0,R_0\}$. It is easy to verify that $\{S_0/\tr S_0,R_0\}$ is a feasible solution to the SDP (\ref{prime F+}) of $f^{+}(\cN,\tr S_0)$. Therefore,  $$f^{+}(\cN,\b(\cN)) \le \tr (S_0/\tr S_0)=1.$$

It is not difficult to see that $f^{+}(\cN, m)$  monotonically decreases when $m$ increases.
Thus, for any $2^r>\beta(\cN)$, we  have $f^{+}(\cN, 2^r)<1$. Then, by Eq. (\ref{n-shot error}), it is clear that the corresponding $n$-shot error probability
 $\e_n$  will go to one exponentially fast as $n$ increases.
 Hence,  $C_\b(\cN)=\log \beta(\cN)$ is a strong converse bound for the NS$\cap$PPT-assisted classical capacity of $\cN$.

Furthermore, let us choose $R_{AB}=\|J_{\cN}^{T_B}\|_\infty \1_{AB}$ and $S_B=\|J_{\cN}^{T_B}\|_\infty \1_{B}$. It is clear that $\{R_{AB},S_{B}\}$ is a feasible solution to the SDP (\ref{sdp_beta}) of $\b(\cN)$, which means that $\b(\cN)\le d_{B}\|J_{\cN}^{T_B}\|_\infty$.
\end{proof}
\textbf{Remark}
$C_\b$ has some remarkable properties. For example, it is
additive: $C_\b(\cN_1\ox\cN_2)=C_\b(\cN_1)+C_\b(\cN_2)$ for different
quantum channels $\cN_1$ and $\cN_2$. This can be proved by utilizing
semidefinite programming duality.

With similar techniques, we are going to show another SDP strong converse bound for the classical capacity of a general quantum channel.
\begin{theorem}\label{SDP strong converse 2}
For a quantum channel $\cN$, we  derive the following strong converse bound for the NS$\cap$PPT assisted classical capacity, i.e.,
$$C(\cN)\le C_{\rm{NS \cap PPT}}(\cN)\le C_\zeta(\cN)= \log\zeta (\cN)$$  
with
 \begin{equation}\begin{split}\label{dual zeta}
\zeta (\cN)= \min &\tr S_{B} \\
 \rm{ s.t. }&  V_{AB}\ge J_{\cN}, -\1_A\ox S_B\le V_{AB}^{T_B}\le \1_A\ox S_B
\end{split}\end{equation}

And if the communication rate exceeds $C_{ \zeta}(\cN)$, the error probability will go to one exponentially fast as the number of channel uses increase.
\end{theorem}
\begin{proof}
We first introduce the following SDP to estimate the optimal success probability:
\begin{equation}\label{prime hat F+}
\begin{split}
\widetilde f^{+}(\cN,m)=\min &\ \tr S_B \\
  \rm{ s.t. } &\ V_{AB}\ge J_{\cN},  \\
 & -m\1_A\ox S_B\le V_{AB}^{T_B}\le m\1_A\ox S_B.
\end{split}\end{equation}
Similar to Proposition \ref{F A UB}, we can prove that
\begin{equation}
 f_{\text{NS}\cap\text{PPT}}(\cN^{\ox n}, m^n)\le \widetilde f^{+}(\cN, m)^n.
\end{equation}

Then, when the communication rate exceeds $C_{ \zeta}(\cN)$, we can use the technique in Theorem \ref{SDP strong converse} to prove that the error probability will go to one exponentially fast as the number of channel uses increase.
%
\end{proof}

As an example, we first apply our bounds to the qudit noiseless channel. In this case, the bounds are tight and strictly smaller than the entanglement-assisted classical capacity.
\begin{proposition}
For the qudit noiseless channel $I_d(\rho)=\rho$, it holds that
\begin{equation}
C(I_d)=C_\beta(I_d)=C_\zeta(I_d)=\log d<2\log d=C_{\rm{E}}(I_d).
\end{equation}
\end{proposition}
\begin{proof}
It is clear that $C(I_d)\ge \log d$.
 By the fact that $\|J_{I_d}^{T_B}\|_{\infty}=1$, it is easy to see that  $C_\beta(I_d)\le \log d\|J_{I_d}^{T_B}\|_{\infty}=\log d$. Similarly, we also have $C_\zeta(I_d)\le \log d$. And $C_{\rm{E}}(I_d)=2\log d$ is due to the  superdense coding \cite{Bennett1992}.
\end{proof}

\subsection{Amplitude damping channel}
For the amplitude damping channel $\cN_{\gamma}^{AD}=\sum_{i=0}^1 E_i\cdot E_i^\dag$ $(0\le \gamma\leq 1)$
with $E_0=\ketbra{0}{0}+\sqrt{1-\gamma}\ketbra{1}{1}$ and $E_1=\sqrt{\gamma}\ketbra{0}{1}$, the Holevo capacity $\chi(\cN_{\gamma}^{AD})$ is given in \cite{Giovannetti2005}. However, its classical capacity remains unknown. The only known nontrivial  and meaningful upper bound for the classical capacity of the amplitude damping channel was established in \cite{Brandao2011c}.
As an application of theorems \ref{SDP strong converse} and \ref{SDP strong converse 2}, we show a strong converse bound for the classical capacity of the qubit amplitude damping channel. Remarkably, our bound improves the best previously known upper bound  \cite{Brandao2011c}.

\begin{theorem}
For amplitude damping channel  $\cN_{\gamma}^{AD}$, 
$$C_{\rm{NS \cap PPT}}(\cN_{\gamma}^{AD})\le C_{\zeta}(\cN_{\gamma}^{AD})= C_{\beta}(\cN_{\gamma}^{AD})  =\log (1+\sqrt {1-\gamma}).$$
As a consequence,
$$C(\cN_{\gamma}^{AD}) \le\log (1+\sqrt {1-\gamma}).$$
\end{theorem} 
\begin{proof}
Suppose that $$S_B=\frac{\sqrt{1-\gamma}+1+\gamma}{2}\proj{0}+\frac{\sqrt{1-\gamma}+1-\gamma}{2}\proj{1}$$
and
$$V_{AB}=J_{\gamma}^{AD}+(\sqrt{1-\gamma}-1+\gamma)\proj{v}
$$
with $\ket v=\frac{1}{\sqrt 2}(\ket{00}+\ket{11})$.

It is clear that $V_{AB}\ge J_{\gamma}^{AD}$. Moreover, it is easy to see that
$$
\1_A\ox S_B-V_{AB}^{T_B}=\frac{\sqrt{1-\gamma}+1-\gamma}{2}(\ket{01}-\ket{10})(\bra{01}-\bra{10})\ge 0$$
and
$\1_A\ox S_B+V_{AB}^{T_B} 
= (\sqrt{1-\gamma}+1+\gamma)\proj{00}
+(\sqrt{1-\gamma}+1-\gamma)\proj{11} 
+ \frac{\sqrt{1-\gamma}+1-\gamma}{2}(\proj{01}+\ketbra{01}{10}+\ketbra{10}{01})
+\frac{\sqrt{1-\gamma}+1+3\gamma}{2}\proj{10}
\ge0
$.

Therefore, $\{S_B, V_{AB}\}$ is a feasible solution to SDP (\ref{dual zeta}), which means that
$$C_{\zeta}(\cN_{\gamma}^{AD}) \le \log\tr S_B=\log (1+\sqrt {1-\gamma}).$$
One can also use the dual SDP of $C_{\beta}$ to show that $C_{\beta}(\cN_{\gamma}^{AD})\ge\log (1+\sqrt {1-\gamma})$. Hence, we have that $C_{\zeta}(\cN_{\gamma}^{AD})=\log (1+\sqrt {1-\gamma})$. 

Similarly, it can also be calculated that $C_{\beta}(\cN_{\gamma}^{AD})=\log (1+\sqrt {1-\gamma})$.
\end{proof}
\textbf{Remark:}
It is worth noting that our bound is strictly smaller than the entanglement-assisted capacity when  $\gamma\le 0.75$ as shown in the following FIG. \ref{adCE}.
We further compare our bound with the previous upper bound \cite{Brandao2011c} and lower bound  \cite{Giovannetti2005} in FIG. \ref{ad}. 
The authors of \cite{Giovannetti2005} showed that $$C(\cN_\gamma^{AD})\ge \max_{0\le p\le 1} \{H_2[(1-\gamma)p]-H_2[(1+\sqrt{1-4(1-\gamma)\gamma p^2})/2]\},$$ where $H_2$ is the binary entropy.
It is clear that our bound provides a tighter bound to the classical capacity than the previous bound \cite{Brandao2011c}.

\begin{figure}[htbp]
\centering
\begin{minipage}{.48\textwidth}
 \centering{\includegraphics[width=6.22cm]{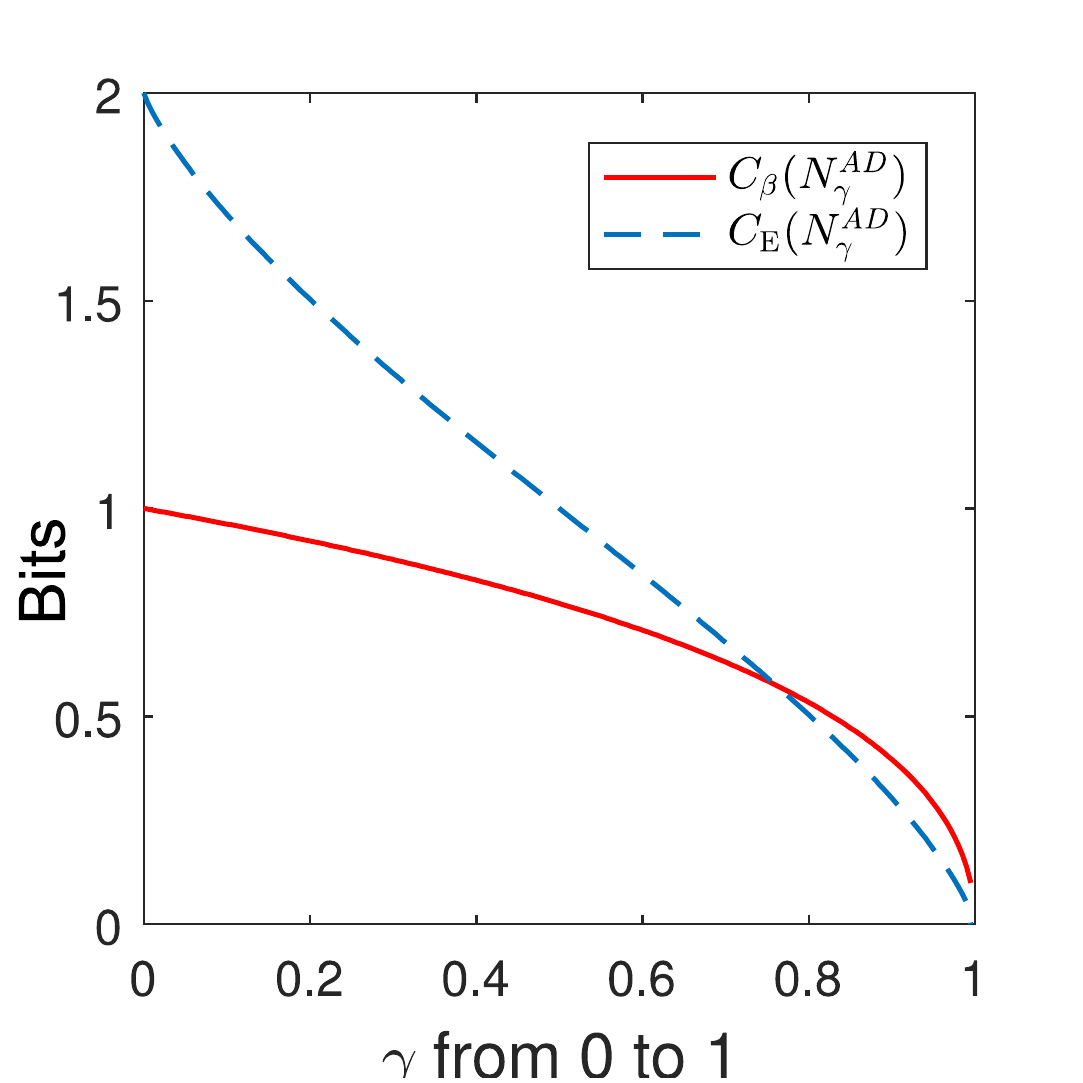}}
  \caption{The  solid line depicts $C_\beta(\cN_{\gamma}^{AD})$ while the dashed line depicts $C_{\rm E}(\cN_{\gamma}^{AD})$. It is worth noting that  $C_\beta(\cN_{\gamma}^{AD})$ is strictly smaller than $C_{\rm E}(\cN_{\gamma}^{AD})$ for any $\gamma\le0.75$.}
\label{adCE}
\end{minipage}
\hspace{0.1cm}
\begin{minipage}{.48\textwidth}
 \centering{\includegraphics[width=6.3cm]{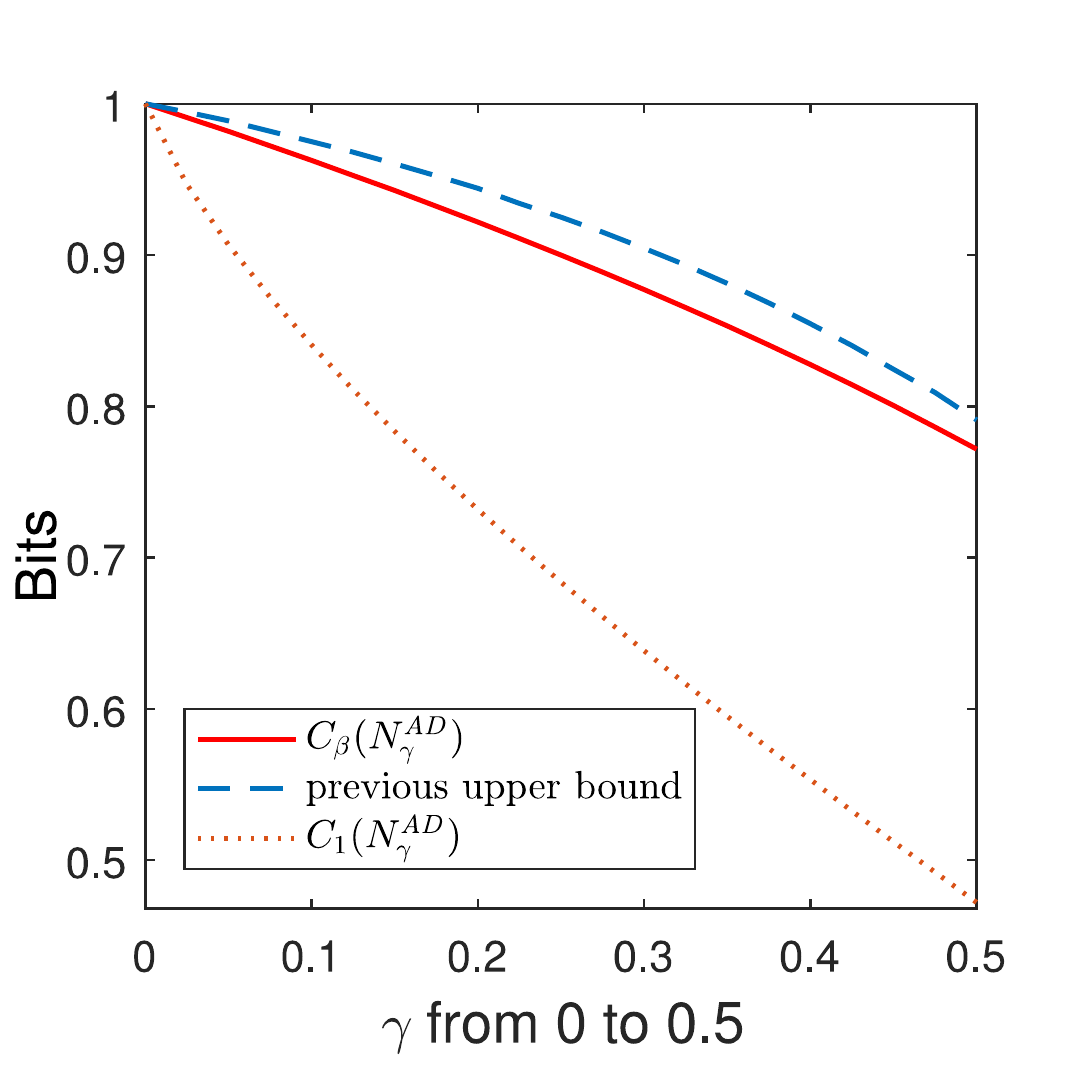}}
  \caption{The  solid line depicts  $C_\beta(\cN_{\gamma}^{AD})$,  the dashed line depicts the previous bound of  $C(\cN_{\gamma}^{AD})$   \cite{Brandao2011c}, and the dotted line depicts the lower bound \cite{Giovannetti2005}.
 Our bound is tighter than the previous bound in \cite{Brandao2011c}.}
\label{ad}
\end{minipage}
\end{figure}

%

\subsection{Strong converse property for a new class of quantum channels}
In  \cite{Wang2016f}, a class of qutrit-to-qutrit channels was introduced to show the separation between quantum Lov\'asz number and entanglement-assisted zero-error classical capacity. 
It turns out that this class of channels also has strong converse property for classical or private communication.
To be specific, the channel from register $A$ to $B$ is given by $\cN_{\alpha}(\rho)=E_0\rho E_0^{\dagger}+E_1\rho E_1^{\dagger}$ $(0<\alpha\le \pi/4)$
 with 
\begin{align*}
E_0 = \sin \alpha\ketbra{0}{1}+\ketbra{1}{2}, E_1=\cos\alpha\ketbra{2}{1}+\ketbra{1}{0}.
\end{align*}
It follows that the complementary channel of $\cN_\a$ is 
$\cN_\a^c(\rho)=\sum_{i=0}^{2}F_i\rho F_i^{\dagger}$ with 
\begin{align*}
F_0 = \sin \alpha\ketbra{0}{1}, F_1 = \ketbra{0}{2}+\ketbra{1}{0}, 
F_2 = \cos\alpha\ketbra{1}{1}.
\end{align*}

\begin{proposition}\label{C N_A}
For  $\cN_{\alpha}$ $(0<\alpha\le \pi/4)$,
we have that
\begin{equation*}
C(\cN_{\a})= C_{\rm{NS}\cap \rm{PPT}}(\cN_{\a})=C_{\beta}(\cN_\a)=1.
\end{equation*}
\end{proposition}
\begin{proof}
Suppose the $Z_B=\sin^2\a\proj 0 +\cos^2\a\proj 2+\proj{1}$ and 
\begin{align*}
R&_{AB}=\proj{01}+\proj{11}+\proj{21}
+\sin^2\a(\proj{10}+\proj{20})\\
&+\cos^2\a(\proj{02}+\proj{12})
+\sin\a\cos\a(\ketbra{02}{20}+\ketbra{20}{02}).
\end{align*}
It is easy to check that 
\begin{align*}
-R_{AB}\le J_{\cN_\a}^{T_B}\le R_{AB} \text{ and }
 -\1_A\ox Z_B\le R_{AB}^{T_B}\le \1_A\ox Z_B,
 \end{align*}
 where $J_{\cN_\a}$ is the Choi-Jamio\l{}kowski matrix of $\cN_{\alpha}$. 
 
 Therefore,  $\{Z_B, R_{AB}\}$ is a feasible solution of  SDP (\ref{sdp_beta}) of $\b(\cN_{\a})$, which means that $$\b(\cN_{\a})\le \tr Z_B=2.$$
Noticing that we can use input $\proj 0$ and $\proj{1}$ to transmit two messages via $\cN$, we can conclude that
$$C(\cN_{\a})=C_{\rm{NS}\cap \rm{PPT}}(\cN_{\a})= 1.$$
\end{proof}

\begin{remark}
In \cite{Wang2016f}, the entanglement-assisted  capacity of $\cN_\a$ is shown to be 
$$C_{\rm{E}}(\cN_\a)=2.$$
Therefore, for $\cN_\a$ $(0<\alpha\le \pi/4)$, our bound $C_\b$ is strictly smaller than the entanglement-assisted capacity. In this case, we also note that 
$C_{\b}(\cN_\a) < C_{\zeta}(\cN_\a)$.
However, it remains unknown whether $C_{\b}$ is always smaller than or equal to $C_{\zeta}$.

Furthermore, it is easy to see that $\cN_\a$ is neither an entanglement-breaking channel nor a Hadamard channel. Note also that $\cN_\a$ is not belong to the three classes in \cite{Koenig2009}, for which the strong converse for classical capacity has been established. Thus, our results show a new class of quantum channels which satisfy the strong converse property for classical capacity.
\end{remark}

Moreover, we find that the strong converse property also holds for the private classical capacity  \cite{Devetak2005a,Cai2004} of $\cN_\a$. Note that private capacity requires that no information leaked to the environment and is usually called $P(\cN)$. Recently, several converse bounds for private communication were established in \cite{Wilde2016c,Pirandola2015b,Wilde2016d,Takeoka2014,Christandl2016}. 
\begin{proposition}
 The private capacity of $\cN_\a$ is exactly one bit, i.e.,
$P(\cN_\a)=1$.
In particular,
$$Q(\cN_\a)\le \log(1+\cos \a)<1=P(\cN_\a)=C(\cN_\a)=\frac{1}{2}C_{E}(\cN_\a). $$
\end{proposition}
\begin{proof}
On one hand, it is easy to see that $P(\cN_\a)\le C(\cN_\a)=C_\beta(\cN_\a)=1$.

On the other hand, Alice can choose two input states $\ket{\psi_0}=\ket{1}$ and  $\ket{\psi_1}=\cos\a\ket 0+\sin\a\ket 2$, then the corresponding output states Bob received are
\begin{align*}
\cN_\a(\proj{\psi_0})&=\sin\a^2\proj{0} +\cos\a^2\proj{2},\\
\cN_\a(\proj{\psi_1})&= \proj{1}.
\end{align*}
 It is clear that Bob can perfectly distinguish these two output states. Meanwhile, the corresponding outputs of the complementary channel $\cN_\a^c$ are same, i.e.,
 \begin{align*}
\cN_\a^c(\proj{\psi_0})=\cN_\a^c(\proj{\psi_1})=\sin\a^2\proj{0} +\cos\a^2\proj{1},
\end{align*}
which means that the environment obtain zero information during the communication.

 Applying the SDP bound of the quantum capacity in \cite{Wang2016a}, the quantum capacity of $\cN_\a$ is strictly smaller than $\log(1+\cos \a)$.
\end{proof}

Our result establishes the strong converse property  for  both the classical and private capacities of $\cN_\a$. For the classical capacity, such a property was previously only known for classical channels,  identity channel, entanglement-breaking channels, Hadamard channels and particular covariant quantum channels \cite{Wilde2014a,Koenig2009}. For the private capacity,  such a property was previously only known for generalized dephasing channels and quantum erasure channels \cite{Wilde2016c}. Moreover, our result also shows a simple example of the distinction between the private and the quantum capacities, which were discussed in \cite{Horodecki2005,Leung2014}.

\section{Zero-error capacity}
\label{NS PPT zero error}
 While ordinary information theory focuses on sending messages with asymptotically vanishing errors \cite{Shannon1948}, Shannon also investigated this problem in the zero-error setting and described the zero-error capacity of a channel as the maximum rate at which it can be used to transmit information with zero probability of error \cite{Shannon1956}.  Recently the zero-error information theory has been studied in the quantum setting and many new interesting phenomena have been found 
 \cite{Duan2008a, Cubitt2010,Duan2013, Cubitt2011a, Leung2012, Cubitt2012, Duan2015a, Duan2015}.

The one-shot zero-error capacity of a quantum channel $\cN$ is the maximum number of inputs such that the receiver can perfectly distinguish the corresponding output states.  Cubitt et al. \cite{Cubitt2011} first introduced  the zero-error communication via classical channels assisted by classical no-signalling correlations. Recently,  no-signalling-assisted   zero-error communication over quantum channels was introduced in \cite{Duan2016}.

Using the expression (\ref{C1 NSPPT}) for our one-shot $\e$-error capacity, we are going to show a formula for the one-shot zero-error classical capacity assisted by NS (or NS$\cap$PPT)  codes.
\begin{theorem}
The one-shot zero-error classical capacity (quantified as messages) of
  $\cN$ assisted by $\rm{NS \cap PPT}$ codes is given by
\begin{equation}\begin{split}
  \label{eq:Upsilon PPT}
  M_{0, \rm{NS \cap PPT}} (\cN) = \max &\ \tr S_A \\
   \rm{ s.t.} &\ 0 \leq U_{AB} \leq S_A \ox \1_B, \\
   &\ \tr_A U_{AB} = \1_B, \\
         &\ \tr J_{\cN}(S_A\ox\1_B-U_{AB}) = 0, \\
         &\ 0 \leq U_{AB}^{T_{B}} \leq S_A \ox \1_B \ (\rm{PPT}).
\end{split}\end{equation}
To obtain $M_{0, \rm{NS}}(\cN)$, one only needs to remove the PPT constraint.
By the regularization,
 the $\O$-assisted zero-error classical capacity is
$$C_{0,\O}(\cN)    = \sup_{n\geq 1} \frac1n \log M_{0, \O}(\cN^{\ox n}).$$
\end{theorem}
\begin{proof}
When $\e=0$, it is easy to see that
\begin{equation}\begin{split}
C_{\text{\rm{NS}}\cap\text{\rm{PPT}}}^{(1)}(\cN,0)= -\log \min &\  \eta  \\
 \rm{ s.t. }  &\  0\le F_{AB}\le \rho_A\otimes \1_B, \\
 &\ \tr\rho_A=1, \tr_{A}F_{AB}=\eta\1_B, \\
&\ \tr J_{\cN}F_{AB}\ge 1, \\
&\ 0\le F_{AB}^{T_B}\le \rho_A\otimes \1_B.\\
\end{split}\end{equation}

Then, assuming that $x=1/\eta$, $U_{AB}=xF_{AB}$ and $S_A=x\rho_A$, we have that
\begin{equation}\label{M NS PPT 1}
\begin{split}
M_{0, \rm{NS \cap PPT}}(\cN)&=2^{C_{\text{\rm{NS}}\cap\text{\rm{PPT}}}^{(1)}(\cN,0)} \\
& =\max \ \tr S_A \\  
&\phantom{==}\rm{ s.t. }  \  0\le U_{AB}\le S_A\otimes \1_B, \\
&\phantom{==  \rm{ s.t. }} \tr_{A}U_{AB}=\1_B, \\
&\phantom{==  \rm{ s.t. } } \tr J_{\cN}U_{AB}\ge \tr S_A, \\
&\phantom{== \rm{ s.t. } } \ 0\le U_{AB}^{T_B}\le S_A\otimes \1_B.\\
\end{split}\end{equation}

By the fact that $\tr S_A=\tr J_\cN(S_A\ox\1_B)$, the third constraint in Eq. (\ref{M NS PPT 1}) is equivalent to 
$\tr J_{\cN}(S_A\ox\1_B-U_{AB})\le 0$. Noticing that $S_A\otimes \1_B-U_{AB}\ge 0$, we can simplify Eq. (\ref{M NS PPT 1}) to
\begin{equation}\label{M NS PPT 2}
\begin{split}
M_{0, \rm{NS \cap PPT}}(\cN)= \max &\ \tr S_A  \\
  \rm{ s.t. }  &\  0\le U_{AB}\le S_A\otimes \1_B, \\
  &\ \tr_{A}U_{AB}=\1_B, \\
&\ \tr J_{\cN}(S_A\ox\1_B-U_{AB}) = 0, \\
&\  0\le U_{AB}^{T_B}\le S_A\otimes \1_B.\\
\end{split}\end{equation}
\end{proof}
\begin{remark}
It is worth noting that $M_{0, \rm{NS}}(\cN)$ coincides with the no-signalling assisted zero-error capacity in \cite{Duan2016}. Also, it can be proved that $M_{0, \rm{NS\cap PPT}}$  also depends only on the non-commutative bipartite graph \cite{Duan2016} of $\cN$.
\end{remark}

A natural application of $M_{0, \O}(\cN)$  is to upper bound the one-shot zero-error capacity, i.e.,
$$M_{0}(\cN) \le M_{0, \rm{NS\cap PPT}}(\cN)\le M_{0, {\rm{NS}}}(\cN).$$
It is known that computing the one-shot  zero-error capacity of a quantum channel is QMA-complete \cite{Beigi2007}. However, our bounds can be efficiently solved by semidefinite programming. 
To the best of our knowledge, 
for a general quantum channel $\cN=\sum_iE_i\cdot E_i^\dagger$,
the best known bound of the one-shot zero-error  capacity is the naive form of the Lov\'asz number $\vartheta(\cN)$ in \cite{Duan2013}, i.e.,
$$\vartheta(\cN)=\vartheta(\cS)=\max \{\|\1+T \|_{\infty}: \ T\in \cS^{\perp}, \1+T\ge 0 \},$$
where $\cS=\text{span}\{E_j^\dagger E_k\}$ is the non-commutative graph of $\cN$.

In the next Proposition, we show that $M_{0, \rm{NS\cap PPT}}(\cN)$ can be strictly smaller than $\vartheta(\cN)$ for some quantum channel $\cN$. This implies that $M_{0, \rm{NS\cap PPT}}(\cN)$ can provide a more accurate estimation of the one-shot zero-error capacity of some general quantum channels.
\begin{proposition}
For  $\cN_{\alpha}$ $(0<\alpha\le \pi/4)$,
$$M_{0, \rm{NS\cap PPT}}(\cN_\a)<\vartheta(\cN_\a).$$
\end{proposition}
\begin{proof}
One one hand,  one can also use the prime and dual SDPs of $M_{0, \rm{NS\cap PPT}}$ to prove $M_{0, \rm{NS\cap PPT}}(\cN_\a)\le 2$.
Indeed, this is also easy to see by Proposition \ref{C N_A}.

On the other hand, 
we are going to prove $\vartheta(\cN_\a)\ge 1+\cos^{-2}\a $. Suppose that $T_0=-\proj{0}+\cos^{-2}\a \proj{1}+(1-\cos^{-2}\a )\proj{2}$. It is clear that $T_0 \in S^\perp$ and $\1+T_0 \ge 0$. Thus, $$\vartheta(\cN_\a)\ge \|\1+T_0\|_\infty = \bra 1 (\1+T_0)\ket 1 =  1+\cos^{-2}\a>2.$$
\end{proof}

For this class of quantum channels, it is worth noting that the private zero-error capacity is also one bit while its quantum zero-error capacity is strictly smaller than one qubit, i.e.,
$Q_0(\cN_\a)<1=P_0(\cN_\a)=C_0(\cN_\a)$.
This shows a difference between the quantum and the private capacities of a quantum
channel in the zero-error setting,
which relates to the work about maximum privacy without coherence in the zero-error case \cite{Leung2015}.

\section{Conclusions and Discussions}\label{conclusion}
In summary,  we have established fundamental limits for classical communication over quantum channels by considering general codes with NS constraint or NS$\cap$PPT constraint. New SDP bounds for classical communication under both finite blocklength and asymptotic settings are obtained in this work.

We first study the finite blocklength regime. By imposing both no-signalling and PPT-preserving constraints, we have obtained the optimal success probabilities of transmitting classical information assisted by NS and NS$\cap$PPT codes. Based on this, we have also derived the 
one-shot $\e$-error NS-assisted and NS$\cap$PPT-assisted capacities. In particular, all of these one-shot characterizations are in the form of semidefinite programs.  The one-shot NS-assisted and NS$\cap$PPT-assisted)  $\e$-error capacities provide an improved finite blocklength estimation of the classical communication than the previous quantum hypothesis testing converse bounds  in \cite{Matthews2014}.  Moreover, for classical channels, the one-shot NS-assisted and NS$\cap$PPT-assisted  $\e$-error capacities are equal to the linear program for the Polyanskiy-Poor-Verd\'u converse bound \cite{Polyanskiy2010,Matthews2012}, thus giving an alternative proof of that result.
Furthermore, in the asymptotic regime, we derive two SDP strong converse bounds of the classical capacity of a general quantum channel, which are efficiently computable and can be strictly smaller than the entanglement-assisted capacity.
As an example, we have shown an improved upper bound on the classical capacity of the qubit amplitude damping channel. 
Moreover, we have proved that the strong converse property holds for both classical and private capacities for a new class of quantum channels. This result may help us deepen the understanding of the limit ability of a quantum channel to transmit classical information. 

Finally, we apply our results to the study of zero-error capacity. To be specific, based on our SDPs of optimal success probability, we have derived the  one-shot NS-assisted (or NS$\cap$PPT-assisted )  zero-error capacity. Our result of NS-assisted capacity  provides an alternative derivation for the NS-assisted zero-error capacity in \cite{Duan2016}. Moreover, the one-shot NS$\cap$PPT-assisted zero-error capacity  also provide some insights in quantum zero-error information theory.

It would be interesting to study the asymptotic capacity $C_{\rm{NS}\cap \rm{PPT}}$ using such techniques as quantum hypothesis testing. Maybe it also has a single-letter formula similar to entanglement-assisted classical capacity. Perhaps one can obtain tighter converse bounds via the study of $C_{\rm{NS}\cap \rm{PPT}}$.
Another direction is to further tighten the one-shot and strong converse bounds by involving the separable constraint  \cite{Harrow2015}.  It would also be interesting to study how to implement the no-signalling and PPT-preserving codes.

\section*{Acknowledgement}
The authors are grateful to Jens Eisert, Stefano Pirandola, Mark M. Wilde and Dong Yang for their helpful suggestions.
Xin Wang would like to thank Mario Berta, Hao-Chung Cheng, Omar Fawzi, William Matthews, Volkher B. Scholz, Marco Tomamichel and Andreas Winter for helpful discussions. 
We also thank the Associate Editor and the anonymous referees for valuable comments that helped improving the paper.
This work was partly supported by the Australian Research Council under Grant Nos. DP120103776 and FT120100449.

\bibliographystyle{IEEEtran}
\bibliography{Bib}

\begin{thebibliography}{10}
\providecommand{\url}[1]{#1}
\csname url@samestyle\endcsname
\providecommand{\newblock}{\relax}
\providecommand{\bibinfo}[2]{#2}
\providecommand{\BIBentrySTDinterwordspacing}{\spaceskip=0pt\relax}
\providecommand{\BIBentryALTinterwordstretchfactor}{4}
\providecommand{\BIBentryALTinterwordspacing}{\spaceskip=\fontdimen2\font plus
\BIBentryALTinterwordstretchfactor\fontdimen3\font minus
  \fontdimen4\font\relax}
\providecommand{\BIBforeignlanguage}[2]{{%
\expandafter\ifx\csname l@#1\endcsname\relax
\typeout{** WARNING: IEEEtran.bst: No hyphenation pattern has been}%
\typeout{** loaded for the language `#1'. Using the pattern for}%
\typeout{** the default language instead.}%
\else
\language=\csname l@#1\endcsname
\fi
#2}}
\providecommand{\BIBdecl}{\relax}
\BIBdecl

\bibitem{Wang2017b}
X. Wang, W. Xie, and R. Duan, “Semidefinite programming converse bounds for classical communication over quantum channels,” \emph{2017 IEEE International Symposium on Information Theory (ISIT)}. pp. 1728–1732, 2017.

\bibitem{Holevo1973}
A.~S. Holevo, ``{Bounds for the quantity of information transmitted by a
  quantum communication channel},'' \emph{Problemy Peredachi Informatsii
  (Problems of Information Transmission)}, vol.~9, no.~3, pp. 3--11, 1973.

\bibitem{Holevo1998}
------, ``{The capacity of the quantum channel with general signal states},''
  \emph{IEEE Transactions on Information Theory}, vol.~44, no.~1, pp. 269--273,
  1998.

\bibitem{Schumacher1997}
B.~Schumacher and M.~D. Westmoreland, ``{Sending classical information via
  noisy quantum channels},'' \emph{Physical Review A}, vol.~56, no.~1, p. 131,
  1997.

\bibitem{King2003}
C.~King, ``{The capacity of the quantum depolarizing channel},'' \emph{IEEE
  Transactions on Information Theory}, vol.~49, no.~1, pp. 221--229, 2003.

\bibitem{Bennett1997}
C.~H. Bennett, D.~P. DiVincenzo, and J.~A. Smolin, ``{Capacities of quantum
  erasure channels},'' \emph{Physical Review Letters}, vol.~78, no.~16, p.
  3217, 1997.

\bibitem{King2002}
C.~King, ``{Additivity for unital qubit channels},'' \emph{Journal of
  Mathematical Physics}, vol.~43, no.~10, pp. 4641--4653, 2002.

\bibitem{Amosov2000}
G.~G. Amosov, A.~S. Holevo, and R.~F. Werner, ``{On Some Additivity Problems in
  Quantum Information Theory},'' \emph{Problemy Peredachi Informatsii},
  vol.~36, no.~4, pp. 25--34, 2000.

\bibitem{Datta2006}
N.~Datta, A.~S. Holevo, and Y.~Suhov, ``{Additivity for transpose depolarizing
  channels},'' \emph{International Journal of Quantum Information}, vol.~4,
  no.~01, pp. 85--98, 2006.

\bibitem{Fukuda2005}
M.~Fukuda, ``{Extending additivity from symmetric to asymmetric channels},''
  \emph{Journal of Physics A: Mathematical and General}, vol.~38, no.~45, p.
  L753, 2005.

\bibitem{Konig2012}
R.~Konig, S.~Wehner, and J.~Wullschleger, ``{Unconditional security from noisy
  quantum storage},'' \emph{IEEE Transactions on Information Theory}, vol.~58,
  no.~3, pp. 1962--1984, 2012.

\bibitem{Hastings2008a}
M.~B. Hastings, ``{Superadditivity of communication capacity using entangled
  inputs},'' \emph{Nature Physics}, vol.~5, no.~4, pp. 255--257, apr 2009.

\bibitem{Beigi2007}
S.~Beigi and P.~W. Shor, ``{On the complexity of computing zero-error and
  Holevo capacity of quantum channels},'' \emph{arXiv:0709.2090}, 2007.

\bibitem{Wolfowitz1978}
J.~Wolfowitz, ``{Coding theorems of information theory},'' \emph{Mathematics of
  Computation}, 1978.

\bibitem{Ogawa1999}
T.~Ogawa and H.~Nagaoka, ``{Strong converse to the quantum channel coding
  theorem},'' \emph{IEEE Transactions on Information Theory}, vol.~45, no.~7,
  pp. 2486--2489, 1999.

\bibitem{Winter1999}
A.~Winter, ``{Coding theorem and strong converse for quantum channels},''
  \emph{IEEE Transactions on Information Theory}, vol.~45, no.~7, pp.
  2481--2485, 1999.

\bibitem{Koenig2009}
R.~Koenig and S.~Wehner, ``{A strong converse for classical channel coding
  using entangled inputs},'' \emph{Physical Review Letters}, vol. 103, no.~7,
  p. 70504, 2009.

\bibitem{Wilde2013a}
M.~M. Wilde and A.~Winter, ``{Strong converse for the classical capacity of the
  pure-loss bosonic channel},'' \emph{Problems of Information Transmission},
  vol.~50, no.~2, pp. 117--132, 2013.

\bibitem{Wilde2014a}
M.~M. Wilde, A.~Winter, and D.~Yang, ``{Strong converse for the classical
  capacity of entanglement-breaking and Hadamard channels via a sandwiched
  R{\'{e}}nyi relative entropy},'' \emph{Communications in Mathematical
  Physics}, vol. 331, no.~2, pp. 593--622, 2014.

\bibitem{Polyanskiy2010}
Y.~Polyanskiy, H.~V. Poor, and S.~Verd{\'{u}}, ``{Channel coding rate in the
  finite blocklength regime},'' \emph{IEEE Transactions on Information Theory},
  vol.~56, no.~5, pp. 2307--2359, 2010.

\bibitem{Hayashi2009}
M.~Hayashi, ``{Information spectrum approach to second-order coding rate in
  channel coding},'' \emph{IEEE Transactions on Information Theory}, vol.~55,
  no.~11, pp. 4947--4966, 2009.

\bibitem{Matthews2012}
W.~Matthews, ``{A linear program for the finite block length converse of
  polyanskiy-poor-verd{\'{u}} via nonsignaling codes},'' \emph{IEEE
  Transactions on Information Theory}, vol.~58, no.~12, pp. 7036--7044, 2012.

\bibitem{Matthews2014}
W.~Matthews and S.~Wehner, ``{Finite blocklength converse bounds for quantum
  channels},'' \emph{IEEE Transactions on Information Theory}, vol.~60, no.~11,
  pp. 7317--7329, 2014.

\bibitem{Wang2012}
L.~Wang and R.~Renner, ``{One-shot classical-quantum capacity and hypothesis
  testing},'' \emph{Physical Review Letters}, vol. 108, no.~20, p. 200501,
  2012.

\bibitem{Renes2011}
J.~M. Renes and R.~Renner, ``{Noisy channel coding via privacy amplification
  and information reconciliation},'' \emph{IEEE Transactions on Information
  Theory}, vol.~57, no.~11, pp. 7377--7385, 2011.

\bibitem{Tomamichel2013a}
M.~Tomamichel and M.~Hayashi, ``{A hierarchy of information quantities for
  finite block length analysis of quantum tasks},'' \emph{IEEE Transactions on
  Information Theory}, vol.~59, no.~11, pp. 7693--7710, 2013.

\bibitem{Berta2011a}
M.~Berta, M.~Christandl, and R.~Renner, ``{The quantum reverse Shannon theorem
  based on one-shot information theory},'' \emph{Communications in Mathematical
  Physics}, vol. 306, no.~3, pp. 579--615, 2011.

\bibitem{Leung2015c}
D.~Leung and W.~Matthews, ``{On the power of PPT-preserving and non-signalling
  codes},'' \emph{IEEE Transactions on Information Theory}, vol.~61, no.~8, pp.
  4486--4499, 2015.

\bibitem{Tomamichel2015}
M.~Tomamichel and V.~Y.~F. Tan, ``{Second-order asymptotics for the classical
  capacity of image-additive quantum channels},'' \emph{Communications in
  Mathematical Physics}, vol. 338, no.~1, pp. 103--137, 2015.

\bibitem{Beigi2015}
S.~Beigi, N.~Datta, and F.~Leditzky, ``{Decoding quantum information via the
  Petz recovery map},'' \emph{Journal of Mathematical Physics}, vol.~57, no.~8,
  p. 082203, aug 2016.

\bibitem{Tomamichel2015b}
M.~Tomamichel, \emph{{Quantum Information Processing with Finite Resources:
  Mathematical Foundations}}.\hskip 1em plus 0.5em minus 0.4em\relax Springer,
  2015, vol.~5.

\bibitem{Tomamichel2016}
M.~Tomamichel, M.~Berta, and J.~M. Renes, ``{Quantum coding with finite
  resources},'' \emph{Nature Communications}, vol.~7, p. 11419, 2016.

\bibitem{Fang2017}
K.~Fang, X.~Wang, M.~Tomamichel, and R.~Duan, ``{Non-asymptotic entanglement
  distillation},'' \emph{arXiv:1706.06221}, pp. 1--20, jun 2017.

\bibitem{Cheng2017b}
H.-C. Cheng and M.-H. Hsieh, ``{Moderate Deviation Analysis for
  Classical-Quantum Channels and Quantum Hypothesis Testing},''
  \emph{arXiv:1701.03195}, 2017.

\bibitem{Chubb2017}
C.~T. Chubb, V.~Y.~F. Tan, and M.~Tomamichel, ``{Moderate deviation analysis
  for classical communication over quantum channels},''
  \emph{arXiv:1701.03114}, 2017.

\bibitem{Mosonyi2009a}
M.~Mosonyi and N.~Datta, ``{Generalized relative entropies and the capacity of
  classical-quantum channels},'' \emph{Journal of Mathematical physics},
  vol.~50, no.~7, p. 72104, 2009.

\bibitem{Rains2001}
E.~M. Rains, ``{A semidefinite program for distillable entanglement},''
  \emph{IEEE Transactions on Information Theory}, vol.~47, no.~7, pp.
  2921--2933, 2001.

\bibitem{Beckman2001}
D.~Beckman, D.~Gottesman, M.~A. Nielsen, and J.~Preskill, ``{Causal and
  localizable quantum operations},'' \emph{Physical Review A}, vol.~64, no.~5,
  p. 52309, 2001.

\bibitem{Eggeling2002a}
T.~Eggeling, D.~Schlingemann, and R.~F. Werner, ``{Semicausal operations are
  semilocalizable},'' \emph{EPL (Europhysics Letters)}, vol.~57, no.~6, p. 782,
  2002.

\bibitem{Piani2006}
M.~Piani, M.~Horodecki, P.~Horodecki, and R.~Horodecki, ``{Properties of
  quantum non-signaling boxes},'' \emph{Physical Review A}, vol.~74, no.~1, p.
  12305, 2006.

\bibitem{Oreshkov2012}
O.~Oreshkov, F.~Costa, and {\v{C}}.~Brukner, ``{Quantum correlations with no
  causal order},'' \emph{Nature Communications}, vol.~3, p. 1092, 2012.

\bibitem{Duan2016}
R.~Duan and A.~Winter, ``{No-signalling-assisted zero-error capacity of quantum
  channels and an information theoretic interpretation of the Lov{\'{a}}sz
  number},'' \emph{IEEE Transactions on Information Theory}, vol.~62, no.~2,
  pp. 891--914, 2016.

\bibitem{Datta2013c}
N.~Datta and M.-H. Hsieh, ``{One-shot entanglement-assisted quantum and
  classical communication},'' \emph{IEEE Transactions on Information Theory},
  vol.~59, no.~3, pp. 1929--1939, 2013.

\bibitem{Bennett1999}
C.~H. Bennett, P.~W. Shor, J.~A. Smolin, and A.~V. Thapliyal,
  ``{Entanglement-assisted classical capacity of noisy quantum channels},''
  \emph{Physical Review Letters}, vol.~83, no.~15, p. 3081, 1999.

\bibitem{Brandao2011c}
F.~G. S.~L. Brandao, J.~Eisert, M.~Horodecki, and D.~Yang, ``{Entangled inputs
  cannot make imperfect quantum channels perfect},'' \emph{Physical Review
  Letters}, vol. 106, no.~23, p. 230502, 2011.

\bibitem{Shannon1956}
C.~E. Shannon, ``{The zero error capacity of a noisy channel},'' \emph{IRE
  Trans. Inf. Theory}, vol.~2, no.~3, pp. 8--19, 1956.

\bibitem{Duan2013}
R.~Duan, S.~Severini, and A.~Winter, ``{Zero-error communication via quantum
  channels, noncommutative graphs, and a quantum Lov{\'{a}}sz number},''
  \emph{IEEE Transactions on Information Theory}, vol.~59, no.~2, pp.
  1164--1174, 2013.

\bibitem{Jamiokowski1972}
A.~Jamio{\l}kowski, ``{Linear transformations which preserve trace and positive
  semidefiniteness of operators},'' \emph{Reports on Mathematical Physics},
  vol.~3, no.~4, pp. 275--278, 1972.

\bibitem{Choi1975}
M.-D. Choi, ``{Completely positive linear maps on complex matrices},''
  \emph{Linear algebra and its applications}, vol.~10, no.~3, pp. 285--290,
  1975.

\bibitem{Vandenberghe1996}
L.~Vandenberghe and S.~Boyd, ``{Semidefinite programming},'' \emph{SIAM
  Review}, vol.~38, no.~1, pp. 49--95, 1996.

\bibitem{Wang2016}
X.~Wang and R.~Duan, ``{Improved semidefinite programming upper bound on
  distillable entanglement},'' \emph{Physical Review A}, vol.~94, no.~5, p.
  050301, nov 2016.

\bibitem{Harrow2015}
A.~W. Harrow, A.~Natarajan, and X.~Wu, ``{An Improved Semidefinite Programming
  Hierarchy for Testing Entanglement},'' \emph{Communications in Mathematical
  Physics}, vol. 352, no.~3, pp. 881--904, jun 2017.

\bibitem{Wang2016c}
X.~Wang and R.~Duan, ``{Nonadditivity of Rains' bound for distillable
  entanglement},'' \emph{Physical Review A}, vol.~95, no.~6, p. 062322, jun
  2017.

\bibitem{Li2017}
Y.~Li, X.~Wang, and R.~Duan, ``{Indistinguishability of bipartite states by
  positive-partial-transpose operations in the many-copy scenario},''
  \emph{Physical Review A}, vol.~95, no.~5, p. 052346, may 2017.

\bibitem{Berta2015}
M.~Berta and M.~Tomamichel, ``{The Fidelity of Recovery Is Multiplicative},''
  \emph{IEEE Transactions on Information Theory}, vol.~62, no.~4, pp.
  1758--1763, apr 2016.

\bibitem{Xie2017}
W.~Xie, K.~Fang, X.~Wang, and R.~Duan, ``{Approximate broadcasting of quantum
  correlations},'' \emph{Physical Review A}, vol.~96, p. 022302, aug 2017.

\bibitem{Khachiyan1980}
L.~G. Khachiyan, ``{Polynomial algorithms in linear programming},'' \emph{USSR
  Computational Mathematics and Mathematical Physics}, vol.~20, no.~1, pp.
  53--72, 1980.

\bibitem{Grant2008}
M.~Grant and S.~Boyd, ``{CVX: Matlab software for disciplined convex
  programming},'' 2008.

\bibitem{NathanielJohnston2016}
{Nathaniel Johnston}, ``{QETLAB: A MATLAB toolbox for quantum entanglement,
  version 0.9},'' 2016.

\bibitem{Watrous2011b}
J.~Watrous, \emph{{Theory of quantum information}}.\hskip 1em plus 0.5em minus
  0.4em\relax University of Waterloo, 2011.

\bibitem{Chiribella2008}
G.~Chiribella, G.~M. D'Ariano, and P.~Perinotti, ``{Transforming quantum
  operations: Quantum supermaps},'' \emph{EPL (Europhysics Letters)}, vol.~83,
  no.~3, p. 30004, 2008.

\bibitem{Bennett1992}
C.~H. Bennett and S.~J. Wiesner, ``{Communication via one-and two-particle
  operators on Einstein-Podolsky-Rosen states},'' \emph{Physical Review
  Letters}, vol.~69, no.~20, p. 2881, 1992.

\bibitem{Bennett1993}
C.~H. Bennett, G.~Brassard, C.~Cr{\'{e}}peau, R.~Jozsa, A.~Peres, and W.~K.
  Wootters, ``{Teleporting an unknown quantum state via dual classical and
  Einstein-Podolsky-Rosen channels},'' \emph{Physical Review Letters}, vol.~70,
  no.~13, p. 1895, 1993.

\bibitem{Giovannetti2005}
V.~Giovannetti and R.~Fazio, ``{Information-capacity description of spin-chain
  correlations},'' \emph{Physical Review A}, vol.~71, no.~3, p. 32314, 2005.

\bibitem{Wang2016f}
X.~Wang and R.~Duan, ``{Separation between quantum Lov{\'{a}}sz number and
  entanglement-assisted zero-error classical capacity},''
 \emph{IEEE Transactions on Information Theory}, vol.~64, no.~3,
  pp. 1454--1460, 2018.

\bibitem{Devetak2005a}
I.~Devetak, ``{The private classical capacity and quantum capacity of a quantum
  channel},'' \emph{IEEE Transactions on Information Theory}, vol.~51, no.~1,
  pp. 44--55, 2005.

\bibitem{Cai2004}
N.~Cai, A.~Winter, and R.~W. Yeung, ``{Quantum privacy and quantum wiretap
  channels},'' \emph{Problems of Information Transmission}, vol.~40, no.~4, pp.
  318--336, 2004.

\bibitem{Wilde2016c}
M.~M. Wilde, M.~Tomamichel, and M.~Berta, ``{Converse bounds for private
  communication over quantum channels},'' \emph{IEEE Transactions on
  Information Theory}, vol.~63, no.~3, pp. 1792--1817, feb 2016.

\bibitem{Pirandola2015b}
S.~Pirandola, R.~Laurenza, C.~Ottaviani, and L.~Banchi, ``{Fundamental limits
  of repeaterless quantum communications},'' \emph{Nature Communications},
  vol.~8, p. 15043, apr 2017.

\bibitem{Wilde2016d}
M.~M. Wilde, ``{Squashed entanglement and approximate private states},''
  \emph{Quantum Information Processing}, pp. 1--18, 2016.

\bibitem{Takeoka2014}
M.~Takeoka, S.~Guha, and M.~M. Wilde, ``{The squashed entanglement of a quantum
  channel},'' \emph{IEEE Transactions on Information Theory}, vol.~60, no.~8,
  pp. 4987--4998, 2014.

\bibitem{Christandl2016}
M.~Christandl and A.~M{\"{u}}ller-Hermes, ``{Relative Entropy Bounds on
  Quantum, Private and Repeater Capacities},'' \emph{Communications in
  Mathematical Physics}, vol. 353, no.~2, pp. 821--852, apr 2016.

\bibitem{Wang2016a}
X.~Wang and R.~Duan, ``{A semidefinite programming upper bound of quantum
  capacity},'' in \emph{2016 IEEE International Symposium on Information Theory
  (ISIT)}, jul
  2016, pp. 1690--1694.

\bibitem{Horodecki2005}
K.~Horodecki, M.~Horodecki, P.~Horodecki, and J.~Oppenheim, ``{Secure key from
  bound entanglement},'' \emph{Physical Review Letters}, vol.~94, no.~16, p.
  160502, 2005.

\bibitem{Leung2014}
D.~Leung, K.~Li, G.~Smith, and J.~A. Smolin, ``{Maximal privacy without
  coherence},'' \emph{Physical Review Letters}, vol. 113, no.~3, p. 30502,
  2014.

\bibitem{Shannon1948}
C.~E. Shannon, ``{A mathematical theory of communication},'' \emph{ACM
  SIGMOBILE Mobile Computing and Communications Review}, vol.~5, no.~1, pp.
  3--55, 1948.

\bibitem{Duan2008a}
R.~Duan and Y.~Shi, ``{Entanglement between two uses of a noisy multipartite
  quantum channel enables perfect transmission of classical information},''
  \emph{Physical Review Letters}, vol. 101, no.~2, p. 20501, 2008.

\bibitem{Cubitt2010}
T.~S. Cubitt, D.~Leung, W.~Matthews, and A.~Winter, ``{Improving zero-error
  classical communication with entanglement},'' \emph{Physical Review Letters},
  vol. 104, no.~23, p. 230503, 2010.

\bibitem{Cubitt2011a}
T.~S. Cubitt, J.~Chen, and A.~W. Harrow, ``{Superactivation of the asymptotic
  zero-error classical capacity of a quantum channel},'' \emph{IEEE
  Transactions on Information Theory}, vol.~57, no.~12, pp. 8114--8126, 2011.

\bibitem{Leung2012}
D.~Leung, L.~Mancinska, W.~Matthews, M.~Ozols, and A.~Roy,
  ``\BIBforeignlanguage{English}{{Entanglement can increase asymptotic rates of
  zero-error classical communication over classical channels}},''
  \emph{\BIBforeignlanguage{English}{Communications in Mathematical Physics}},
  vol. 311, no.~1, pp. 97--111, 2012.

\bibitem{Cubitt2012}
T.~S. Cubitt and G.~Smith, ``{An extreme form of superactivation for quantum
  zero-error capacities},'' \emph{IEEE Transactions on Information Theory},
  vol.~58, no.~3, pp. 1953--1961, 2012.

\bibitem{Duan2015a}
R.~Duan and X.~Wang, ``{Activated zero-error classical capacity of quantum
  channels in the presence of quantum no-signalling correlations},''
  \emph{arXiv:1510.05437}, 2015.

\bibitem{Duan2015}
R.~Duan, S.~Severini, and A.~Winter, ``{On zero-error communication via quantum
  channels in the presence of noiseless feedback},'' \emph{IEEE Transactions on
  Information Theory}, vol.~62, no.~9, pp. 5260--5277, 2016.

\bibitem{Cubitt2011}
T.~S. Cubitt, D.~Leung, W.~Matthews, and A.~Winter, ``{Zero-error channel
  capacity and simulation assisted by non-local correlations},'' \emph{IEEE
  Transactions on Information Theory}, vol.~57, no.~8, pp. 5509--5523, 2011.

\bibitem{Leung2015}
D.~Leung and N.~Yu, ``{Maximum privacy without coherence, zero-error},''
  \emph{Journal of Mathematical Physics}, vol.~57, no.~9, 2016.

\end{thebibliography}

\end{document}